\documentclass[12pt,english]{article}
\usepackage{libertine}
\usepackage[T1]{fontenc}

\usepackage{tabularx}
\usepackage{rotating}
\usepackage{booktabs}
\usepackage{graphicx}
\usepackage{subfig}
\usepackage{caption}
\usepackage{enumerate, ushort}
\usepackage{verbatim}
\usepackage{amsmath}
\usepackage{amssymb,bbm}
\usepackage{amsthm,amsfonts,amscd,graphics} 
\usepackage{graphicx}
\usepackage{epsfig}  
\usepackage{color}
\usepackage{setspace}
\usepackage[multiple]{footmisc}
\usepackage{enumitem}
\usepackage{adjustbox}
\usepackage{multirow}

\usepackage{setspace}
\usepackage[authoryear]{natbib}
\usepackage[unicode=true,pdfusetitle,
 bookmarks=true,bookmarksnumbered=false,bookmarksopen=false,
 breaklinks=false,pdfborder={0 0 1},backref=page,colorlinks=true,citecolor=darkblue,linkcolor=maroon,urlcolor=maroon]
 {hyperref}
\usepackage[nameinlink,noabbrev,capitalize]{cleveref}

\makeatletter

\usepackage{tikz}
\usetikzlibrary{shapes, arrows,positioning,automata,backgrounds}
\tikzstyle{terminator} = [rectangle, draw, text centered, rounded corners, minimum height=2em]
\tikzstyle{pool} = [rectangle, draw, text centered, rounded corners, minimum height=6em]
\tikzstyle{box} = [rectangle, draw, text centered, rounded corners, gray, minimum height=8em, minimum width=14em]
\tikzstyle{connector} = [draw, -latex']

\usepackage{color}
\usepackage{transparent}
\usepackage{subfig}
\PassOptionsToPackage{usenames,dvipsnames,svgnames}{xcolor}
\usepackage{natbib}
\usepackage{epstopdf}
\usepackage{adjustbox}

\newtheorem{lem}{Lemma}
\newtheorem{thm}{Theorem}
\newtheorem{prop}{Proposition}
\newtheorem{corol}{Corollary}
\newtheorem{assum}{Assumption}
\newtheorem{defn}{Definition}
\newtheorem{example}{Example}

\definecolor{darkgreen}{HTML}{006400}
\definecolor{navyblue}{HTML}{000080}
\definecolor{purple}{HTML}{800080}
\definecolor{darkgray}{HTML}{505050}
\definecolor{brown}{HTML}{A52A2A}
\definecolor{teal}{HTML}{008080}
\definecolor{readablered}{HTML}{B22222}
\definecolor{darkblue}{rgb}{0.0,0,.6}
\definecolor{maroon}{rgb}{0.68,0,0}
\definecolor{darkgreen}{rgb}{0,0.369,0.086}
\definecolor{gray}{rgb}{.5,.5,.5}
\definecolor{darkred}{rgb}{.6,0,0}
\definecolor{darkgreen}{rgb}{0,.6,0}
\definecolor{aliceblue}{rgb}{0.94, 0.97, 1.0}
\definecolor{bluegray}{rgb}{0.4, 0.6, 0.8}
\definecolor{shadecolor}{rgb}{.94,.97,1}

\newcommand*\diff{\mathop{}\!\mathrm{d}}

\newcommand{\question}[1]{} % omit questions

   \renewcommand{\hat}{\widehat}
   \usepackage[a4paper, margin=1in]{geometry}
        \usepackage{booktabs}
        \usepackage{amsmath, amssymb, amsfonts}
        \usepackage{graphicx}
        \usepackage{siunitx}
        \usepackage{caption}
        \usepackage{tabularx}
        \usepackage{xcolor}
        \usepackage{hyperref}
        \usepackage[utf8]{inputenc}
        \usepackage[T1]{fontenc}
        \usepackage{booktabs}

\title{Good Data and Bad Data:\\ The Welfare Effects of Price Discrimination}

\author{Maryam Farboodi\thanks{MIT Sloan, NBER and CEPR. Email: farboodi@mit.edu.} \and Nima Haghpanah\thanks{Yale School of Management. Email: nima.haghpanah@gmail.com.} \and Ali Shourideh\thanks{Tepper School of Business, Carnegie Mellon University. Email: ashourid@andrew.cmu.edu.}}
\date{\today{\thanks{We thank Nageeb Ali, Dirk Bergemann, Ben Brooks, Alessandro Bonatti, Roberto Corrao, Stephen Morris,  Jonathan Parker, Maryam Saeedi, Roland Strausz, Alexander Wolitzky, Jidong Zhou, and audiences at various seminars and conferences for their thoughts and comments.}}}

\begin{document}
%\doublespacing
\onehalfspacing

\maketitle  
%\vspace{-13mm}
%\begin{center}
%    \large{\href{https://www.andrew.cmu.edu/user/ashourid/papers/FHS01.pdf}{Click here for the latest version.}}
%\end{center}

%\vspace{7mm}

 \thispagestyle{empty} 

\begin{abstract}
%We ask how a monopolist’s use of consumer data for price discrimination affects welfare.  To answer this question, we develop a model of market segmentation subject to residual uncertainty. We give a complete characterization of when data usage is monotonically good or bad for welfare or when the effect is non-monotone. The characterization consists of a reduction of the problem to one with only two demand curves, and a condition for the two-demand-curves case that highlights that information affects welfare in three distinct ways.   When information has a non-monotone effect on welfare, we provide tight bounds on its effects and identify the best direction for providing a information locally. These results provide insights into when data usage for price discrimination should be allowed.

We study how a monopolist’s use of consumer data for price discrimination affects welfare. To answer this question, we develop a model of market segmentation subject to residual uncertainty. We fully characterize when data usage monotonically increases or decreases welfare or when the effect is non-monotone. The characterization reduces the problem to one with only two demand curves, and gives a condition for the two-demand-curves case that highlights that information affects welfare in three distinct ways. In the non-monotone case, we provide tight bounds on the welfare effects of information and identify the best local direction for providing additional information.
\end{abstract}

\newpage
%%%%%%%%%%%%%%%%%%%%%%%%%%%%%%
\section{Introduction} \label{sec:intro}
The rise of big data technologies, allowing firms to collect detailed consumer data to estimate their willingness to pay, has reignited the longstanding debate on the welfare implications of price discrimination. The prevalence of these practices has raised  concerns among policymakers about big tech’s exploitation of consumer data.  A letter that followed a Senate hearing on May 2, 2024, reflects these concerns:
\begin{quote}
    As more consumers shop online, large tech platforms have access to vast stores of personal data [\ldots] that can be exploited by corporations to set prices based on the time of day, location, or even the electronic device used by a consumer.\footnote{See the hearing's \href{https://www.banking.senate.gov/newsroom/majority/brown-demands-answers-on-amazon-and-walmarts-use-of-so-called-dynamic-pricing}{follow-up letter to Amazon} for details.}
\end{quote}
The Federal Trade Commission also recently issued an order seeking information from several companies about their ``segmentation solutions'' that categorize consumers based on location, demographics, and credit history and set different prices for the same good or service.\footnote{See \href{https://www.ftc.gov/system/files/ftc_gov/pdf/sp6b_order_surv_pricing.pdf} {the order} for details.}

% A significant difficulty in regulating data collection practices is that it is close to impossible to perfectly monitor and control how firms use consumer data, which makes highly targeted regulation impractical.  Often the relevant question is if data collection should be permitted, without knowing how much information the firm already has nor how much additional information it might be able to collect. It is clear, however, that if additional consumer data collection is {guaranteed} to be beneficial (harmful) \emph{regardless} of the firm's existing or additional information, the it should (should not) be permitted.  As a first step towards informed policy design, we seek to characterize when such guarantees can be provided.

We study the welfare effects of price discrimination through the lens of information economics.
%To answer this question, we develop a model to study endogenous market segmentation by a monopolist subject to residual uncertainty.  
There is a given set of consumer types, each associated with a downward-sloping demand curve representing a population of consumers with heterogeneous willingness to pay.  
%We analyze segmentations that arise from different information structures, each allowing the seller to segment the market and charge non-uniform prices across segments.
The seller has access to some information structure that maps types to signal realizations, allowing her to segment the market and charge a profit-maximizing price for each segment.  
%Importantly, the signal is only informative of consumer types and cannot be used to distinguish different consumers with the same type.  
Importantly, the information is only about consumer types and cannot distinguish between consumers of the same type. 
Even if the seller perfectly observes types, she faces residual uncertainty in the form of a downward-sloping demand curve for each type and cannot implement first-degree price discrimination.
This residual uncertainty reflects the practical limitations, legal and technological, that sellers face in perfectly predicting individual willingness to pay.  

By providing a framework to study all market segmentations based on consumer types, our work bridges the classical and modern approaches to price discrimination.   The classical approach to this problem, pioneered by \citet{P20}'s foundational work, starts with the same primitives but, translated to our setting, compares only the two extreme cases where the seller is either fully informed or fully uninformed about the types.  Because this literature  compares two extremes, it only provides sufficient conditions for price discrimination to affect welfare positively or negatively, whereas our focus on the smooth class of all information structures allows us to obtain sharp results and complete characterizations. On the other hand, the modern approach to this problem, pioneered by the seminal work of \citet{BBM15}, studies all market segmentations in a setting where the seller has access to signals about consumer \emph{willingness to pay}, which may overstate the seller’s ability to extract surplus in scenarios with constraints on market segmentation.  Our framework merges these two approaches and allows us to study all market segmentations that are subject to residual uncertainty.

Within our framework, we ask two questions.  First, does refining the seller's segmentation monotonically increase or decrease welfare (a convex combination of consumer and producer surplus), or is the effect non-monotone?
In the first two cases, information is ``per se'' good or bad: Simply knowing the fact that the seller is collecting information is enough to specify whether its effect is positive or negative, without the need for knowing \emph{what} the information is.  This distinction is useful because it is often difficult to perfectly monitor and control how firms use consumer data, so it is desirable to know when doing so is necessary.  When the effect is non-monotone, information has the \emph{potential} to be beneficial or harmful, and further analysis might be needed.  Our second question delves deeper into the non-monotone case: How good or bad can information be, and what are the best and the worst {ways to provide additional information?}

%we examine how additional information impacts welfare through a pair of opposing properties: ``monotonically bad'' and ``monotonically good'' information. Information is ``monotonically bad'' if every refinement of any segmentation reduces weighted surplus --- a convex combination of consumer and producer surplus.  Conversely, information is ``monotonically good'' if weighted surplus is higher for any refinement.  We refer to these two properties as welfare-monotonicity properties.

\paragraph{First result.} We first characterize when welfare is monotonically increasing, monotonically decreasing, or non-monotone in information, using three conditions.  The first condition, which we call partial inclusion, says the demand curves are not too far apart in the sense that the optimal monopoly price of each of them is in the interior of any other's domain of prices.  The second condition is a spanning one that says the set of all demand curves is decomposable into at most two basis demand curves.  If either of the first two conditions do not hold, welfare is not monotone in information.  If they both hold, the third condition gives an expression on the two basis demand curves that specifies whether welfare is monotonically increasing, monotonically decreasing, or non-monotone.

Examining this expression reveals a novel insight, namely that information affects welfare in three ways.  To understand these three effects, consider an example.  Suppose there are two types, type $L$ and type $H$, representing low and high income consumers.  Suppose half of the consumers have type $L$, the other half have type $H$, and if the seller could perfectly price discriminate based on the type, she would choose a higher price for type $H$ consumers.  Without any information, the seller offers some uniform price $p$ to all consumers.

Now suppose the seller can use geography as a proxy for income to price discriminate. There are two locations, each representing a segment, where the first location contains $\tfrac{2}{3}$ of type $L$ consumers and $\tfrac{1}{3}$ of type $H$ consumers, and the second location contains the remaining $\tfrac{1}{3}$ of type $L$ consumers and $\tfrac{2}{3}$ of type $H$ consumers.   The seller lowers the price to some $p_1 < p$ for the first location, and increases the price to some $p_2 > p$ for the second location.  \autoref{tab:the three effects simplified} shows, for each type and each price, the fraction of consumers of that type that face that price, before and after information is provided.
\begin{table}[htbp!]
	\def\arraystretch{1.8}
	\centering
	%     \begin{adjustbox}{width=\textwidth,center}
		% \begin{adjustbox}{center}
			\begin{tabular}{lcccc}
				& & $p_1$ & $p$ & $p_2$\\				
				\cline{3-5}
				\multirow{2}{*}{Before} & Type $L$ & 0 & 1 & 0 \\\cline{3-5}
				& Type $H$ & 0 & 1 & 0 \\\hline \hline
				\multirow{2}{*}{After} & Type $L$ & $\frac{2}{3}$  & 0 & $\frac{1}{3}$ \\\cline{3-5}
				& Type $H$ &$\frac{1}{3}$  & 0 & $\frac{2}{3}$ \\	\cline{3-5}
			\end{tabular}
            \caption{The fraction of consumers of each type facing each price before and after information is provided.}
            \label{tab:the three effects simplified}
		\end{table}

The three effects of information are reflected in the table.  The first effect is the \emph{within-type price change} effect: For each type, information disperses prices, with some consumers facing a higher price and some facing a lower price.  The second effect is the \emph{cross-types price change} effect: The price drop applies asymmetrically, to more type $L$ consumers than type $H$ ones.  The third effect is the \emph{price curvature} effect: The size of the price drop might not be equal to the size of the price increase.  The larger the price drop compared to the price increase, the larger the positive effect of providing information.

Our three effects of information are related to those identified in the classical literature that compares the total surplus of full information to that of no information. 
This literature points out that a sufficient condition for price discrimination to decrease total surplus is that the ``output effect'' is negative, that is, the quantity produced by the seller decreases \citep{R69,Sch81,V85}.  This is because the ``misallocation effect'', that is, the fact that the quantity is sold at non-uniform prices, further distorts the allocation and decreases total surplus.  
Our price-curvature effect corresponds to the output effect because they are both about how price discrimination affects the aggregate decisions of the firm (average price versus total quantity). Similarly, our cross-types price change effect corresponds to the misallocation effect because they are both about how price discrimination treats types non-uniformly.
A distinction, however, is that our effects are expressed in terms of prices (and not quantities), which is crucial for obtaining a unified characterization for arbitrary welfare functions (and not just total surplus).
%\footnote{Suppose a given quantity is sold to a market.  Then to calculate consumer surplus, for example, we need to know how this quantity is divided across different types in that market.  For this, calculating the price is necessary, so (total) quantity is not a sufficient statistic for identifying welfare.}  
A second distinction is our within-type price change effect, which has no counter-part and arises because we study partial information: Because types are not fully separated, even consumers of a given type are affected by price discrimination non-uniformly.  Finally, our characterization specifies how to add up the three effects by giving a single expression that identifies the net effect, whereas the classical argument only specifies sufficient conditions for total surplus to decrease.\footnote{\citet{ACV10} uses a different logic to also obtain sufficient conditions under which price discrimination increases welfare.}

Another insight of our characterization is that price discrimination generally has the potential to be beneficial even without ``opening new markets''.  To elaborate, suppose the optimal monopoly price of type $H$ is higher than the highest price in the support of type $L$ so our partial inclusion condition is violated.  If almost all consumers are of type $H$, the seller will optimally exclude type $L$ consumers.  Then partially separating the two types opens a new market and prompts the seller to lower the price to one segment without increasing it in the other. This is indeed an important force in \citet{BBM15}.  They demonstrate the power of this force by showing that as long as there is any exclusion, some information can achieve full efficiency while giving all the gains to consumers.  
The classical literature, on the other hand, imposes a non-exclusion assumption to ask if price discrimination can be beneficial even if this force is shut down.  Our result shows that it generally can. 
When there is no exclusion, so no new markets are opened,  as long as one of the two other conditions of the result are violated, information does not affect welfare monotonically so there exists some information that increases welfare.  

\paragraph{Second result.} Our second result studies the non-monotone case further.  It provides bounds on welfare gains and losses from information and identifies the best and the worst {directions for providing additional information}.  A key role is played by the value function that specifies the welfare of each posterior belief (once the seller prices optimally for that belief), and the Hessian matrix of this value function.  The largest possible gain of information is proportional to the largest eigenvalue of the Hessian matrix, and the best possible direction to provide additional information is in the direction of the corresponding eigenvector.  The largest possible loss and the worst directions are similarly related to the smallest eigenvalue and its eigenvector.  

Notably, although eigenvalues and eigenvectors of arbitrary matrices are often complex objects that cannot be expressed in closed form, we give closed forms for them in our setting.  This is because the Hessian itself in our setting is a low-rank matrix with a separable form.  In fact, the Hessian has at most two non-zero eigenvalues.  Examining the Hessian and its eigenvalues reveals and generalizes the same three effects of information from our two-type analysis.

The result has a simple intuition.  An observation based on \citet{KaG11} is that information monotonically increases welfare if and only if the value function mentioned above is convex.  So it is natural to expect that if the value function is almost but not exactly convex, then the potential harms of information are small.  Eigenvalues of the Hessian of the value function measure its curvature along their corresponding eigenvectors.  The largest eigenvector represents the direction along which the function is ``most convex'', and providing information along that direction provides the largest value proportional to the curvature of the value function which is measured by the eigenvalue.  A similar argument applies to the lowest eigenvalue and its connection to the largest possible harm of information.\footnote{This intuition more generally applies to Bayesian persuasion problems in which the receiver's action is specified by a first-order condition (for example, \citealp{KCW24}).}

% Ali's paragraph that I rewrote:
%Overall, our analysis highlights the importance of a specific regulatory approach to data availability and collection. Specifically, since welfare benefits or losses of data depend on the details of the markets - bounded by the  eigenvalues of the demand system - regulators should approach this somewhat similar to antitrust and merger approval cases. While more and better quantitative work is needed, we view our second result as providing the theoretical foundations for this approach.

The bounds in our second result provide a \emph{quantitative} analogue to our first, qualitative result, which characterizes when these bounds are both positive, both negative, or have opposite signs. These bounds could therefore inform a quantitative framework to regulate the use of consumer data in price discrimination.  Analogous to the guidelines commonly used in antitrust and merger cases, such a framework could specify thresholds for the potential gains and losses from data usage, and use these bounds to determine whether such practices should be encouraged or discouraged regardless of their form, or whether additional scrutiny, monitoring, or restrictions are warranted. In this latter case, our result on the best and worst ways to provide additional information helps guide the design of appropriate restrictions on data collection and usage.

The rest of the paper is organized as follows. \autoref{sec:litreview} reviews the literature. \autoref{sec:model} describes our model. \autoref{sec:result} defines our welfare-monotonicity properties and gives our first main result, \autoref{thm:NSC}, which provides the necessary and sufficient conditions for information to be monotonically good or bad for welfare, or for the effect to be non-monotone. \autoref{sec:gen} states our second main result, \autoref{thm:bounds}, which provides a tight bound on the impact of information on welfare when the impact is globally non-monotone and identifies the marginally optimal segmentation. \autoref{sec:conc} concludes.  

%and  proves and interprets the two-demands part of the main result. \autoref{sec:method} explains the reduction part of the main result, providing details of the methodology.

%%%%%%%%%%%%%%%%%%%%%%%%%%%%%%%%%%%

\subsection{Related Literature} \label{sec:litreview}
Our work is related to the modern literature on price discrimination that studies all segmentations, pioneered by \citet{BBM15}. Our analysis does not directly apply to the setting of \citet{BBM15} because we make smoothness assumptions to replace the seller's profit-maximization problem with its first-order condition.  Nonetheless, we show that our results are consistent with \citet{BBM15} as we approach their setting within ours. \citet{BBM15} and \citet{KaZ23} identify all possible pairs of consumer and producer surplus. \citet{Yan22} studies how a profit-maximizing intermediary sells segmentations to a producer for price discrimination. \citet{HiV21}, \citet{HS21}, \citet{HS23}, \citet{Ass24}, and \citet{BHW24} study market segmentation with a multi-product seller.  In a general framework, \citet{BHV23} study when consumer surplus is convex or concave in a market's belief, which allows them to characterize when consumers are better or worse off if the market receives additional information about their types. Our conceptual contribution to this literature is to introduce a framework in which segmentations are constrained by some residual uncertainty.\footnote{A working paper version \citet{BBM13} of \citet{BBM15} studies consumer- and producer-surplus pairs in some examples with partial information about types. \citet{StY25} also study a setting in which there is a separation between a consumer's type (which they call characteristic) and willingness to pay.  They study segmentations in which different types are offered the same distribution of prices.  Because we ask whether refining segmentations monotonically affects welfare, our work is also related to those that study information orders in persuasion (\citealp{CuS22}, \citealp{BFK24}).}

Our work also relates to the large classical literature on monopolistic third-degree price discrimination.
%\footnote{Oligopolistic third-degree price discrimination is studied in a literature that is more distant to ours, such as \citet{H89} and \citet{EGK21}.}  
As discussed earlier, the classical approach to this problem compares uniform pricing to full segmentation of a given set of demand curves.  These papers typically either focus on studying total surplus \citep{V85,ACV10,C16} or consumer surplus \citep{C12,AC15}.  Our general approach allows us to characterize welfare-monotonicity for any weighted combination of consumer and producer surplus, including as special cases consumer and total surplus.  Additionally, our framework allows us to study intermediate forms of price discrimination where the types are partially separated.  Most importantly, whereas the papers in the literature give only sufficient conditions for welfare to rise or fall, we give a complete characterization of our welfare-monotonicity conditions.  As we show with examples, our conditions do not imply nor are implied by those in the literature.  On one hand, our conditions are stronger because we seek to characterize a stronger property, but on the other hand, the conditions in the literature are stronger because they are only sufficient for a ranking of the welfare of full segmentation versus uniform pricing.  Nonetheless, the conditions can be qualitatively compared, as we outline in the introduction.

Methodologically, our work builds on the recent literature that uses duality in Bayesian persuasion \citep{Kol18,DM19,DK24,KCW24}.\footnote{See also \citet{IST22}, \citet{SmY24}, \citet{saeedi2020optimal}, and \citet{saeedi2024getting}.}  In particular, we convert our welfare-monotonicity properties to a class of Bayesian persuasion problems that seek to identify when no-information maximizes or minimizes welfare \emph{for all} prior distributions over the given set of demand curves.  Because the seller's profit-maximization problem in our setting is identified by its first-order condition, we take the strong-duality results of \citet{Kol18} and \citet{KCW24} off-the-shelf to solve these Bayesian persuasion problem.  We then show that for no-information to be a solution to the entire class of Bayesian-persuasion problems, our spanning condition must hold.  The main technical result in \citet{KCW24} gives a sufficient optimality condition, the ``twist'' condition, that is reminiscent of, but different from, our spanning condition.  Our characterization identifies an additional condition on two demands capturing the three effects of information, which has no counter-part in \citet{KCW24}, that together give necessary \emph{and sufficient} conditions for our welfare-monotonicity properties.

\section{Model} \label{sec:model}
% Our model consists of a buyer and a seller trading a homogeneous good. The buyer has a type $\theta\in\Theta$ where $\Theta$ is a compact subset of a linear normed vector space. The type is drawn from a full-support probability distribution $\mu \in \Delta \Theta$.  A buyer of type $\theta$ has preferences given by
%\[v(q,\theta) - pq.\]
%where $q$ is the quantity of goods purchased, while $p$ is the price. Facing the price $p$ set by the seller, the buyer's utility maximization gives rise to a 
%demand curve $D(p,\theta)$ that specifies a quantity purchased at price $p$. Since the demand function $D(p,\theta)$ fully determines her preferences, we will use the demand function $D(p,\theta)$ as representing the buyer's preferences\footnote{In \autoref{subsec:unitdemand}, we provide a different interpretation of the model that encompasses the unit demand model of \citet{BBM15}}.

% Alternative with a continuum of consumers to discuss.
A seller produces a product at a constant marginal cost that is normalized to zero. There is a unit mass of unit-demand consumers, where each consumer has a type and a willingness to pay for the product.  The set of types $\Theta$ is a compact subset of a linear normed vector space (which holds trivially if $\Theta$ is finite).  The types are distributed according to a full-support prior distribution $\mu_0 \in \Delta \Theta$. Each type $\theta \in \Theta$ represents a group of consumers with heterogeneous willingness to pay. The demand curve $D(p,\theta)$ specifies the fraction of type $\theta$ consumers whose willingness to pay for the product is at least $p$.\footnote{Our model has other interpretations.  One is that the seller faces a single consumer with a random type drawn from $\mu$. The consumer has multi-unit demands with non-linear utility for quantity and her utility-maximization problem for a given price induces a demand curve for each type.  One can also combine the two interpretations and think about each demand as representing a population of consumers with multi-unit demands.}   We denote the family of demand curves by $\mathcal{D} = \{D(p,\theta)\}_{\theta \in \Theta}$.   The primitive of our model is a pair $(\mathcal{D},\mu_0)$, that is, the family of demand curves and the prior distribution over them.

We make the following assumption on the demand curves throughout the paper.
\begin{assum}\label{ass1}
For all types $\theta\in\Theta$, the demand curve $D(\cdot,\theta): \mathbb{R}_+\rightarrow \mathbb{R}_+$ satisfies the following properties:
\begin{enumerate}
    \item There exists an interval $I(\theta) = \left[\underline{p}(\theta),\overline{p}(\theta)\right]$, $0\leq \underline{p}(\theta) < \overline{p}(\theta) \leq \infty$, such that $D(p,\theta)$ is differentiable  and strictly decreasing in $p$, $D_p(p,\theta) < 0$, for values  $p\in I(\theta)$.  These intervals need not be identical across types.
    \item For values of $p<\underline{p}(\theta)$, $D(p,\theta) = D(\underline{p}(\theta),\theta)$, and for values of $p>\overline{p}(\theta)$, $D(p,\theta) = 0$.
    \item For values of $p\in I(\theta)$, the revenue function $R(p,\theta) = p D(p,\theta)$ is strictly concave.\footnote{This assumption allows us to replace the seller's profit-maximization problem by its first-order condition.  \citet{KCW24} show that the first-order approach remains valid under a slight relaxation of strict concavity, roughly assuming concavity up to a normalization, so our results go through with this weaker assumption.  We state our assumption with strict concavity because it is more straightforward and is commonly used in the literature.}
    \item There exists $p(\theta)\in I(\theta)$ such that $R_p (p(\theta),\theta) =0$.
\end{enumerate}
\end{assum}

The above assumption implies that $p(\theta)$ is the optimal (revenue-maximizing) price set by the seller when she faces only type $\theta$ consumers.

Demand curves that are step functions, each representing a population of unit-demand consumers who all have the same willingness to pay, corresponds to the model in \citet{BBM15} but violates \autoref{ass1} (because such a demand curve is not strictly decreasing over a non-empty interval).
Later we show how our analysis confirms their results by studying demand curves that pointwise approach step functions while maintaining \autoref{ass1}.

A market $\mu \in \Delta \Theta$ is a probability distribution over types.   A segmentation $\sigma \in \Delta \Delta \Theta$ is a Bayes-plausible distribution over markets, that is, $\mathbb{E}_\sigma[\mu] = \mu_0$, and $S(\mu_0)$ denotes the set of all segmentations.  The interpretation is that the seller has access to some information structure that reveals a signal about the type of the buyer.  The seller then forms a posterior $\mu$ after each signal realization and chooses a profit-maximizing price for that posterior.  The seller uses a pricing rule $p: \Delta \Theta \rightarrow \mathbb{R}_+$ that specifies an optimal price for every possible posterior, breaking ties if necessary,
	\begin{align*}
	p(\mu) \in \arg \max_{p \in \mathbb{R}_+}  \int_{\Theta} R(p,\theta) \diff \mu(\theta), \forall \mu \in \Delta \Theta.
\end{align*}

We study the welfare that each segmentation induces.  For this, let 
\begin{align*}
	CS(p,\theta) = \int_{p}^{\overline{p}(\theta)} D(z,\theta) \diff z
\end{align*}
denote the surplus of type $\theta$ consumers from facing price $p$.  The $\alpha$-surplus of price $p$ for type $\theta$ consumers is a weighted average of consumer surplus and producer surplus with weights $\alpha \in (0,1]$ and $1-\alpha$,
\begin{align*}
	V^\alpha(p,\theta) := \alpha \cdot CS(p,\theta) + (1-\alpha)\cdot R(p,\theta).
\end{align*}
The special case of $\alpha = 1$ corresponds to consumer surplus and the special case of $\alpha = \tfrac{1}{2}$ corresponds to total surplus.\footnote{We assume $\alpha > 0$ because our properties of interest are trivial if the entire weight is on producer surplus (we explain why after defining the properties).}   The $\alpha$-surplus of a segmentation is the expectation of the $\alpha$-surplus of all price and type pairs in the segmentation, given by
\begin{align*}
	V^{\alpha}(\sigma) := \int_{\Delta \Theta}  \int_{\Theta} V^{\alpha}(p(\mu),\theta) \diff \mu(\theta)  \diff \sigma(\mu).
\end{align*}

%\section{Main Result:  Surplus-Monotonicity Characterization} 
\section{Welfare-Monotonicity Characterization and Implications} \label{sec:result}
Our first result characterizes when allowing the monopolist to refine segmentations monotonically affects welfare, and when the effect is non-monotone.  We first define these monotonicity properties and then give the result, its proof sketch, interpretation, and examples.

Consider two segmentations $\sigma$ and $\sigma'$.  We say $\sigma$ is a mean-preserving spread of $\sigma'$, and write 
\begin{align*}
\sigma \underset{\text{MPS}}\succeq \sigma',
\end{align*}
if there exists a joint distribution over pairs of markets $\nu \in \Delta(\Delta \Theta \times \Delta \Theta)$ that induces marginals $\sigma$ and $\sigma'$, i.e., $\nu(\cdot,\Delta\Theta) = \sigma(\cdot)$ and $\nu(\Delta\Theta,\cdot) = \sigma'(\cdot)$,  and random markets $(\mu,\mu')$ drawn from $\nu$ satisfy $\mathbb{E}[\mu | \mu'] = \mu'$ almost surely.  We call $\sigma$ a ``refinement'' of $\sigma'$ because it can be obtained by splitting each market in segmentation $\sigma'$ into possibly multiple markets in a mean-preserving way.  
%If $\sigma$ is a mean-preserving spread of $\sigma'$, then we can garble the signals of an information structure that leads to $\sigma$ to obtain $\sigma'$ (as shown by \citealp{blackwell1953equivalent}), so an information structure that corresponds to $\sigma$ Blackwell-dominates that of $\sigma'$.  
The \textit{full-information segmentation}, in which each market $\mu$ in the support of the segmentation has only a single type in its support, is finer than any segmentation, and any segmentation is finer than the \textit{no-information segmentation} that assigns probability 1 to the prior market $\mu_0$.

Given this definition of refinement, we define welfare-monotonicity:

\begin{defn}[{\bf Welfare-monotonicity in information}] Consider a given $(\mathcal{D},\mu_0)$. 
\begin{enumerate}
\item Information is monotonically bad for $\alpha$-surplus, ``$\alpha$-IMB holds'', if
	\begin{align*}
		V^\alpha(\sigma) \leq V^\alpha(\sigma'),\qquad \forall \sigma,\sigma' \in S(\mu_0) \text{ such that } \sigma \underset{\text{MPS}}\succeq \sigma'.
	\end{align*}
	
\item Information is monotonically good for $\alpha$-surplus, ``$\alpha$-IMG holds'', if
\begin{align*}
	V^\alpha(\sigma) \geq V^\alpha(\sigma'),\qquad  \forall \sigma,\sigma' \in S(\mu_0) \text{ such that } \sigma \underset{\text{MPS}}\succeq \sigma'.
\end{align*}
\end{enumerate}
We refer to these properties jointly as ``welfare-monotonicity'' properties.
\end{defn}

% Welfare-monotonicity is related to the question of which segmentation maximizes the weighted surplus over all segmentations, that is, which segmentation solves
% \begin{align*}
% 	&\max_{\sigma \in S(\mu_0)} V^{\alpha}(\sigma).
% \end{align*}
% Monotonically bad information implies that the no-information segmentation solves the above problem, but IMB is a stronger property because it requires that \emph{any} refinement of \emph{any} segmentation leads to a lower weighted surplus (similarly IMB is stronger than optimality of full information).

As stated in the introduction, our motivation for studying these properties is to delineate when the act of information usage by the seller is ``per se'' beneficial or harmful versus when the answer depends on the context (i.e., what information is collected).  This classification is useful because it is often costly to monitor and control what data the firms collect and how they use it, so it is desirable to know when such monitoring and control is necessary.  When welfare-monotonicity holds, depending on its direction, price discrimination should be encouraged or discouraged without the need for any additional context.  When welfare-monotonicity does not hold, further analysis is needed, to which we return later.

Our first main result completely delineates the three possible cases.\footnote{\citet{blackwell1953equivalent} implies that information is monotonically good for producer surplus, which is why we focus on $\alpha > 0$ throughout the paper.}  The result starts by a simplifying step: It shows that the welfare-monotonicity properties are prior-free, allowing us to drop the prior distribution through the rest of the statement. Formally, prior-freeness means that for any two full-support distributions $\mu,\mu'$, $\alpha$-IMG ($\alpha$-IMB) holds for $(\mathcal{D},\mu)$ if and only if $\alpha$-IMG ($\alpha$-IMB) holds for $(\mathcal{D},\mu')$.  

The main content of the result are statements (i) and (ii).  Statement (i) gives a reduction of the problem from any number of demand curves to \emph{binary} families of demand curves.  It gives conditions under which we only need to verify welfare-monotonicity for the binary family $\{D(\cdot,\theta)\}_{\theta \in \{\theta_L,\theta_H\}}$ that consists of the two demand curves with the lowest and the highest optimal monopoly price, $p(\theta_L) = \min\limits_{\theta}p(\theta)$ and $p(\theta_H) = \max\limits_{\theta}p(\theta)$ (ties can be broken arbitrarily).  Statement (ii) provides the characterization for binary families.  

A condition used in the characterization is the \textit{partial inclusion} condition, which means $\underline{p}(\theta') \leq p(\theta) \leq \overline{p}(\theta')$ for all $\theta,\theta' \in \Theta$.  We call this the partial inclusion condition because it says that if type $\theta'$ consumers are offered the optimal price for type $\theta$, some but not all of them will be served.

We now state the result and then go over its two main statements in more detail.

\begin{thm}[{\bf Welfare-monotonicity}] \label{thm:NSC} The welfare-monotonicity properties are prior-free, and we can therefore refer to them as properties of $\mathcal{D}$.  Let $\theta_L,\theta_H \in \Theta$ be the two types with the lowest and the highest optimal monopoly prices in $\mathcal{D}$.  The following two statements hold:
	\begin{enumerate}[label=(\roman*)]
        \item $\alpha$-IMB ($\alpha$-IMG) holds for $\mathcal{D}$ if and only if
		\begin{enumerate}[label=({\Alph*})]
			\item there is partial inclusion, and 
			\item there exist two functions $f_1, f_2: \Theta \rightarrow \mathbb{R}_+$ such that 
            \begin{align*}
            D(p,\theta) = f_1(\theta)D(p,\theta_L) + f_2(\theta) D(p,\theta_H)    
            \end{align*}
            for all $\theta$ and $p \in [p(\theta_L),p(\theta_H)]$, and
        \item $\alpha$-IMB ($\alpha$-IMG) holds for the binary family $\{D(\cdot,\theta)\}_{\theta \in \{\theta_L,\theta_H\}}$.
		\end{enumerate} 
			\item $\alpha$-IMB ($\alpha$-IMG) holds for the binary family $\{D(\cdot,\theta)\}_{\theta \in \{\theta_L,\theta_H\}}$ if and only if there is partial inclusion and 
			%\begin{align*}
			%	V^{\alpha}(p,\theta_H) - V^{\alpha}(p,\theta_L) + \frac{-R_p(p,\theta_L)V^{\alpha}_p(p,\theta_H) + R_p(p,\theta_H)V^{\alpha}_p(p,\theta_L)}{-R_p(p,\theta_L)R_{pp}(p,\theta_H) + R_p(p,\theta_H)R_{pp}(p,\theta_L)}(R_p(p,\theta_L) - R_p(p,\theta_H))
			%\end{align*}
			\begin{align} \hspace{-7mm}
				V^{\alpha}(p,\theta_H) - V^{\alpha}(p,\theta_L) + \frac{V^{\alpha}_p(p,\theta_L)-\frac{R_p(p,\theta_L)}{R_p(p,\theta_H)}V^{\alpha}_p(p,\theta_H)}{R_{pp}(p,\theta_L)-\frac{R_p(p,\theta_L)}{R_p(p,\theta_H)}R_{pp}(p,\theta_H)}(R_p(p,\theta_L) - R_p(p,\theta_H)) \label{eq:main}
			\end{align}
			is decreasing (increasing) on $(p(\theta_L),p(\theta_H))$.
	\end{enumerate}
\end{thm}

%\begin{thm}[{\bf Surplus-Monotonicity}] \label{thm:NSC} Let $\theta_1,\theta_H \in \Theta$ be the two types with the lowest and the highest optimal monopoly prices in $\mathcal{D}$. $\alpha$-IMB ($\alpha$-IMG) holds for $(\mathcal{D},\mu_0)$ if and only if three conditions are satisfied:
%	\begin{enumerate}[label=(\roman*)]
%        \item The following expression is decreasing (increasing) on $[p(\theta_1),p(\theta_H)]$
%        			\begin{align}
%				V^{\alpha}(p,\theta_H) - V^{\alpha}(p,\theta_1) + \frac{-\frac{R_p(p,\theta_1)}{R_p(p,\theta_H)}V^{\alpha}_p(p,\theta_H) + V^{\alpha}_p(p,\theta_1)}{-\frac{R_p(p,\theta_1)}{R_p(p,\theta_H)}R_{pp}(p,\theta_H) + R_{pp}(p,\theta_1)}(R_p(p,\theta_1) - R_p(p,\theta_H)). \label{eq:main}
%			\end{align}
 %       \item There exist two functions $f_1, f_2: \Theta \rightarrow \mathbb{R}_+$ such that $D(p,\theta) = f_1(\theta)D(p,\theta_1) + f_2(\theta) D(p,\theta_H)$ for all $\theta$ and $p \in [p(\theta_1),p(\theta_H)]$.
 %       \item There is partial inclusion.
%	\end{enumerate}
%\end{thm}

Let us now explain in more detail what the theorem says.  Consider statement (i).  Condition (A) of the statement, partial inclusion, says that the monopoly price $p(\theta)$ for any type $\theta$  cannot be less than the lowest price $\underline{p}(\theta')$ in the support of another type $\theta'$ or higher than the largest price $\overline{p}(\theta')$ of type $\theta'$.  Roughly speaking, this partial-inclusion condition means that the demand curves cannot be too far from each other.\footnote{This partial inclusion condition is related to a non-exclusion assumption that is typically made in the classical price discrimination literature.  That assumption only entails one part of our partial inclusion condition, namely that $p(\theta) \leq \overline{p}(\theta')$ for all $\theta,\theta' \in \Theta$.  We review the connection between our finding and those in that literature in \autoref{sec: discussion}.} 
%For example, with linear demands $D(p,\theta) = \theta - p$ where and $\underline{p}(\theta) = 0,\overline{p}(\theta) = \theta$, the optimal monopoly price for each type is $p(\theta) = \frac{\theta}{2}$, and the partial-inclusion condition requires that 
%\begin{align*}
%\frac{\theta}{2} \leq \theta',
%\end{align*}
%for every two types $\theta,\theta'$.  That is, the set of all possible values of $\theta$ must be in some interval $[c,2c]$ for some constant $c \geq 0$.  Notice that this partial-inclusion condition must hold regardless of what $\alpha$ is or which one of the welfare-monotonicity properties we are characterizing.

Condition (B) of statement (i)  is a spanning condition that says that all demand curves in the family must be decomposable into a linear combination of at most two ``basis'' demand curves $D(\cdot,\theta_L),D(\cdot,\theta_H)$ with possibly varying weights, where these basis demand curves are those with the lowest and the highest optimal monopoly price.  The only possible source of heterogeneity among the demand curves is the pair of weights $f_1(\theta)$ and $f_2(\theta)$, so heterogeneity must be reducible to a two-dimensional sufficient statistic.  Notice that this condition must also hold regardless of what the weight $\alpha$ is or which one of the welfare-monotonicity properties we are characterizing.  

Condition (C) of statement (i)  puts additional constraints on the basis demand curves.  The binary family of demand curves that consists only of the two basis demand curves $D(\cdot,\theta_L),D(\cdot,\theta_H)$ must itself satisfy the corresponding welfare-monotonicity property.  This condition therefore reduces the problem of characterizing the welfare-monotonicity properties for an arbitrary class of distributions $\mathcal{D}$ to characterizing them with binary distributions.

Statement (ii) of the theorem characterizes the welfare-monotonicity properties for a binary family of demand curves in terms of partial inclusion and monotonicity of an expression that depends on the weight $\alpha$. The direction of the monotonicity condition depends on which one of the welfare-monotonicity properties we are characterizing.  We re-visit and interpret this expression shortly.

%%%%%%%%%%%%%%%%%%%%%%%%%%%%%%%%%%%%%%%%%%%%%%%%%%%%%%%%%%%%%%%%%%%%%%%%%%%

One intuitive implication of the result, that we use throughout the paper, is that information gets more beneficial as $\alpha$ decreases, that is, more weight is put on the seller's payoff.\footnote{This result can also be directly shown using \citet{blackwell1953equivalent}'s characterization without using ours.  We call it a corollary as a ``sanity check'' of our characterization and to document it for later use.}  

\begin{corol}\label{corol: monotonicity in alpha}
	If $\alpha$-IMG ($\alpha$-IMB) holds, then $\alpha'$-IMG ($\alpha'$-IMB) holds for any $\alpha' \leq \alpha$ ($\alpha' \geq \alpha$).
\end{corol}

In the rest of this section, we visit each of the two main statements of \autoref{thm:NSC} and explain the logic behind them, starting by the case of two demand curves because its analysis reveals a novel insight on the effects of price discrimination.  We then relate our findings to those in the classical and modern literature on price discrimination.

\subsection{Two Demand Curves}\label{sec: two demands}
Let us start with two demand curves and establish and interpret statement (ii) of \autoref{thm:NSC}.  In order to focus on how welfare-monotonicity is related to the monotonicity of \cref{eq:main}, which allows us to provide intuition for the effects of price discrimination, we assume that there is partial inclusion.  In \autoref{sec:method} we illustrate why partial inclusion is in fact necessary for both of the welfare-monotonicity properties.

\subsubsection{Proof Sketch}
It is useful to begin by describing how the result in the binary case is proved.  The proof is based on the concavification approach of  \citet{KaG11}.  To state the result, let us denote by $\mu$ a market in which the probability of type $\theta_H$, the type with the higher monopoly price, is $\mu$.  Consider the weighted-surplus function $W^{\alpha}: [0,1] \rightarrow R$ that specifies the expected $\alpha$-surplus of a market $\mu$,
\begin{align}
	W^\alpha(\mu) &= \mathbb{E}_{\theta \sim \mu}[V^{\alpha}(p(\mu),\theta)],\label{eqn:val2}
\end{align}
where $p(\mu) \in [p(\theta_L),p(\theta_H)]$ is the profit-maximizing price for this  market and is uniquely identified by the first-order condition
\begin{align}
	\mathbb{E}[R_p(p(\mu),\theta)] = 0.\label{eq: two types FOC}
\end{align}

Information is monotonically bad for $\alpha$-surplus \emph{if and only if} $W^\alpha$ is concave.  Indeed, if $W^\alpha$ is concave, then splitting any market $\mu$ in the support of a segmentation into multiple markets with the same mean would only (weakly) decrease weighted surplus.  And if $W^\alpha$ is not concave, that is, if $W^\alpha(\mu)$ is below the concavification of $W^{\alpha}$ for some market $\mu$, then we can take any segmentation with market $\mu$ in its support and split $\mu$ into two markets in a way that increases weighted surplus.  A similar argument shows that information is monotonically good for $\alpha$-surplus if and only if $W^\alpha$ is convex.

To see how the concavity (convexity) of $W^\alpha$ relates to the monotonicity condition in \cref{eq:main}, consider the derivative of $W^\alpha$ with respect to $\mu$,
\begin{align*}
	W_\mu^\alpha(\mu) = & \ - V^{\alpha}(p(\mu),\theta_L)  + V^{\alpha}(p(\mu),\theta_H) + p_\mu(\mu) \mathbb{E}[V_p^{\alpha}(p(\mu),\theta)].
\end{align*}
The weighted surplus function $W^\alpha$ is concave if and only if its derivative is decreasing.  The expression above contains $p(\mu)$ and $p_{\mu}(\mu)$, both of which are defined implicitly given the seller's profit-maximization condition in \cref{eq: two types FOC}. In the proof of the result we eliminate this implicit dependence by observing that because  $p(\mu)$ is increasing in $\mu$, $W_\mu^\alpha$ is a monotone function of $\mu$ if and only if $W_\mu^\alpha(\mu(p))$ is a monotone function of $p$, where $\mu(p)$ is the inverse of the price function, that is, $\mu(p(\mu)) = \mu$.  Indeed, the expression in \cref{eq:main} is exactly the above expression evaluated at $\mu(p)$.

\subsubsection{Interpretation: The Three Effects of Information}\label{sec: the three effects}
To interpret the expression in \cref{eq:main}, let us keep the optimal price function $p(\mu)$ (and its derivatives) implicitly defined and take a second derivative of $W^{\alpha}$.  This second derivative $W_{\mu\mu}^\alpha(\mu)$ represents the value of providing a small amount of information starting from a market $\mu$.  Examining this second derivative reveals that providing information has three effects on weighted surplus that are combined into one formula.  The second derivative of $W^{\alpha}$ is
\begin{align}
	W_{\mu\mu}^\alpha(\mu) = 	(p_{\mu}(\mu))^2 &\ \mathbb{E}\Big[V_{pp}^{\alpha}(p(\mu),\theta)\Big] \nonumber\\
	+  2 p_{\mu}(\mu) & \hspace{4mm} \Big[V^{\alpha}_{p}(p(\mu),\theta_H)  - V^{\alpha}_{p}(p(\mu),\theta_L)\Big] \nonumber\\
	+ p_{\mu\mu}(\mu) &\ \mathbb{E}\Big[V_p^{\alpha}(p(\mu),\theta)\Big].	\label{eq: the three effects}
\end{align}

To understand the three effects, consider what happens if we take a market $\mu$ with optimal price $p$ and split it into two markets $\mu_1 = \mu - \delta$ and $\mu_2 = \mu + \delta$ with optimal prices $p_1 < p < p_2$, each with probability one-half.  \autoref{tab:the three effects} shows, for each type and each price, the fraction of consumers of that type that face that price, before and after information is provided.
\begin{table}[h]
	\def\arraystretch{1.8}
	\centering
	%     \begin{adjustbox}{width=\textwidth,center}
		% \begin{adjustbox}{center}
			\begin{tabular}{lcccc}
				 & & $p_1$ & $p$ & $p_2$\\				
				\cline{3-5}
				\multirow{2}{*}{before} & $\theta_L$ & 0 & 1 & 0 \\\cline{3-5}
				& $\theta_H$ & 0 & 1 & 0 \\\hline \hline
				\multirow{2}{*}{after} & $\theta_L$ & $\frac{1}{2}(1+\frac{\delta}{1-\mu})$  & 0 & $\frac{1}{2}(1-\frac{\delta}{1-\mu})$ \\\cline{3-5}
				& $\theta_H$ &$\frac{1}{2}(1-\frac{\delta}{\mu})$  & 0 & $\frac{1}{2}(1+\frac{\delta}{\mu})$ \\	\cline{3-5}
			\end{tabular}
		\caption{The fraction of consumers of each type facing each price before and after information is provided.}
		\label{tab:the three effects}
	\end{table}

The three effects of information are reflected in \autoref{tab:the three effects}, each represented in one of the terms in \cref{eq: the three effects}.
\begin{enumerate}
	\item The \emph{within-type price change} effect.  Price discrimination disperses prices within each type, where some consumers face a lower price and some others a higher price.  The sign of this effect depends on how $V^{\alpha}$ is affected by a price dispersion, which depends on the curvature of $V^{\alpha}$.  If the function is convex, then this effect is positive, and if it is concave, the effect is negative.  This effect corresponds to the first term in \cref{eq: the three effects}.
	\item The \emph{cross-types price change} effect.  Relatively more type $\theta_L$ consumers face a price drop than do type $\theta_H$ consumers.  In fact, more than half of type $\theta_L$ consumers but less than half of type $\theta_H$ consumers face a price drop.  In the extreme case of full information, all type $\theta_L$ consumers face a price drop and all type $\theta_H$ consumers face a price increase.  This effect is positive if the marginal benefit of decreasing the price for type $\theta_L$ is larger than increasing it for type $\theta_H$, and is negative otherwise.  This effect corresponds to the second term in \cref{eq: the three effects}.
	\item The \emph{price curvature} effect.  The size of price drop $p - p_1$ might not be equal to the size of the price increase $p_2 - p$.  The comparison depends on the curvature of the price function $p(\mu)$.  This effect corresponds to the third term in \cref{eq: the three effects}.  Notice that the seller's profit-maximization condition means that the expected marginal revenue is zero and therefore
    \begin{align*}
    \mathbb{E}\Big[V_p^{\alpha}(p(\mu),\theta)\Big] = \alpha \mathbb{E}\Big[CS_p(p(\mu),\theta)\Big] \leq 0,
    \end{align*}
    and so the sign of the third effect only depends on the curvature of the price function.  This effect is positive if $p(\mu)$ is concave, in which case the price drop is larger than the price increase, benefiting consumers overall.  Similarly, this effect is negative if the price function is convex, so the price increase is larger than the price drop.
\end{enumerate}

%%%%%%%%%%%%%%%%%%%%%%%%%%%%%%%%%%%%%%%%%%%%%%%%%%%

\subsubsection{Sufficient Conditions} \label{subsec:sufCond}
The overall effect of information depends on the \emph{aggregation} of the three effects identified above, and each of the three effects might be positive or negative.  However, if all three terms have the same sign, then their summation has that sign too, and this observation allows us to obtain sufficient conditions for welfare-monotonicity.  The following corollary formalizes this discussion by directly considering each term in \cref{eq: the three effects}.

\begin{corol}\label{cor: sufficient conditions for IMB-IMG}
	$\alpha$-IMB (respectively $\alpha$-IMG) holds if there is partial inclusion and the following hold over the range of prices $(p(\theta_L),p(\theta_H))$.
	\begin{enumerate}
		\item The within-type price change effect is negative (positive):  $V^{\alpha}(p,\theta)$ is concave (convex) for each type $\theta \in \{\theta_L,\theta_H\}$.
		\item The cross-types price change effect is negative (positive):  $V_{p}^{\alpha}(p,\theta_H) \leq (\geq) V_{p}^{\alpha}(p,\theta_L)$.
		\item The price curvature effect is negative (positive):  $p(\mu)$ is convex (concave).  A sufficient condition for this is $R_{ppp}(p,\theta) \geq (\leq)\ 0$ and $R_{pp}(p,\theta_H) \geq (\leq) R_{pp}(p,\theta_L)$.
	\end{enumerate}
\end{corol}

To explain these conditions, let us refer to the demand curve with a lower monopoly price as the ``more elastic'' demand curve.  Indeed, over the interval of prices $[p(\theta_L),p(\theta_H)]$, the elasticity of the demand curve $\theta_L$ is more than 1 (because its marginal revenue is negative), but the elasticity of the demand curve $\theta_H$ is less than 1 (because its marginal revenue is positive).
%\footnote{If $R_p(p,\theta) \leq 0$ then $D(p,\theta) + pD_p(p,\theta) \leq 0$ so the elasticity is $-\frac{pD_p(p,\theta)}{D(p,\theta)} \geq 1$.  A similar argument implies if $R_p(p,\theta) \geq 0$ then $-\frac{pD_p(p,\theta)}{D(p,\theta)} \leq 1$.}  

\autoref{cor: sufficient conditions for IMB-IMG} implies that information is monotonically good for consumer surplus (and all $\alpha$) if
\begin{enumerate}
    \item $R_{ppp}(p,\theta) \leq 0$, that is, the marginal revenue curves are concave,
\end{enumerate}
  and the more elastic demand curve, $\theta_L$, has 
 \begin{enumerate}\setcounter{enumi}{1}
 	\item a higher level, $D(p,\theta_L) \geq D(p,\theta_H)$, and,
 	\item a less concave revenue curve, $R_{pp}(p,\theta_H) \leq R_{pp}(p,\theta_L)$, or equivalently a less steep marginal revenue curve.\footnote{Consider $\alpha = 1$.  The within-group price change effect is positive because consumer surplus is convex in price.  The cross-groups price change effect is positive because $D(p,\theta_L) \geq D(p,\theta_H)$.  The remaining conditions ensure the price curvature effect is positive by \autoref{cor: sufficient conditions for IMB-IMG}.}
 \end{enumerate}
These conditions are shown in \autoref{fig: IMG two demands corollary}. A parallel set of conditions, which we give in Appendix~\ref{app: sufficients}, identify monotonically bad information.
\begin{figure}
	\centering
		
		\begin{tikzpicture}[ultra thick, scale=2.2]
			
			% Draw the x and y axes
			\draw[<->] (5.5,0) node[above]{price} -- (0,0) -- (0,3.1) node[left]{quantity};
			\draw (0,2.9) node[left]{and MR};
			
			%\draw[domain=0.2:1]    plot (4*\x,{2-\x+(.4/\x)+0.3*\x*\x}) node[right,yshift=4]{$D(\cdot,\theta_L)$};
			%\draw[domain=0.2:1]    plot (4*\x,{2.5-1.3*\x+(.03/\x)})node[right,yshift=-4]{$D(\cdot,\theta_H)$};		

			\draw[dashed,domain=1:0.25]    plot (4*\x,{1.4-0.25*\x-1*\x*\x/3+0.3/\x}) node[left,yshift=0]{$D(\cdot,\theta_L)$};
			\draw[domain=1.52/4:3.26/4]    plot (4*\x,{1.4-0.25*\x-1*\x*\x/3+0.3/\x});			
			
			\draw[dashed,domain=1:0.25]    plot (4*\x,{2.3-1.3*\x-1.6*\x*\x/3})node[left,yshift=4]{$D(\cdot,\theta_H)$};		
			\draw[domain=1.52/4:3.26/4]    plot (4*\x,{2.3-1.3*\x-1.6*\x*\x/3});		

			\draw[dashed,red,domain=0.97:0.25]    plot (4*\x,{0.08+.2*\x-1*\x*\x}) node[left,yshift=6]{$R_p(\cdot,\theta_L)$};			
			\draw[red,domain=1.52/4:3.26/4]    plot (4*\x,{0.08+.2*\x-1*\x*\x});

			\draw[dashed,red,domain=0.97:0.25]   plot (4*\x,{2.18-1.85*\x-1*\x*\x}) node[left]{$R_p(\cdot,\theta_H)$} ;			
			\draw[red,domain=1.52/4:3.26/4]   plot (4*\x,{2.18-1.85*\x-1*\x*\x});			
			
			\draw (1.52,0.05) -- (1.52,-0.05) node[below]{$p(\theta_L)$};
			\draw (3.26,0.05) -- (3.26,-0.05) node[below,xshift=-7,yshift=2]{$p(\theta_H)$};
			
			\draw[loosely dotted] (1.52,0.05) -- (1.52,3);
			\draw[loosely dotted] (3.26,0.05) -- (3.26,3);
			
		\end{tikzpicture}
	\caption{Information is monotonically good when both marginal revenue curves are concave and the more elastic demand $\theta_L$ has a higher level and a less steep marginal revenue curve.}
	\label{fig: IMG two demands corollary}
\end{figure}

To conclude our discussion of two demand curves, we apply \autoref{cor: sufficient conditions for IMB-IMG} to an example.  In this example, there is a range of parameters for which all of the effects are positive, implying information is monotonically good, and another range where they are all negative negative, implying information is monotonically bad.

\begin{example}[{\bf Sufficient Condition for $\alpha$-IMG and $\alpha$-IMB}]\label{ex: linear with shifts}
	Consider two demand curves $D(p,\theta_i) = a_i - p + \frac{c_i}{p}$ for $i \in \{1,2\}$ and $a_i,c_i \geq 0$ with supports $[\delta,a_i]$ for small enough $\delta> 0$ (to ensure demands are bounded). Without loss of generality, assume $a_1 \leq a_2$.  Then  $\alpha$-IMG holds for all $\alpha$ if 
	\begin{align*}
		c_1 - c_2 \geq (a_2 - a_1)\frac{a_2}{2}.
	\end{align*}
	$\alpha$-IMB holds for all $\alpha\geq \frac{1}{2}$ if
	\begin{align*}
		c_1 \leq c_2 \leq \frac{a_1^2}{4}.
	\end{align*}
\end{example}

\subsection{Reduction from Many Demand Curves to Two} \label{sec:method}
%%%%%%%%%%%%%%%%%%%%%%%%%%%%%%%%%%
We now sketch the proof of the rest of \autoref{thm:NSC} that shows how to reduce a general family of demand curves to two. The reduction has three main steps. We first show how to transform the problem to an optimization problem and establish prior-freeness. We then show why partial inclusion is necessary for welfare-montonicity and how it enables us to use the first-order approach.  Finally, we use strong duality to solve the optimization problem.

\subsubsection{Step 1: Transformation to an Optimization Problem}\label{subsec: tranformation to optimization}
If $\alpha$-IMB holds for $(\mathcal{D},\mu_0)$, then no-information segmentation maximizes $V^{\alpha}$ over all segmentations in $S(\mu_0)$, but not vice versa.
  We observe, however, that $\alpha$-IMB for $(\mathcal{D},\mu_0)$ is \emph{equivalent} to the property that, for \emph{every} $\mu$, the no-information segmentation is a maximizer of $V^{\alpha}$ over all segmentations in $S(\mu)$.  This equivalence follows from the standard observation that both of these two properties are equivalent to the concavity of the value function $W^{\alpha}: \Delta \Theta \rightarrow R$,
\begin{align}
	W^{\alpha}(\mu) = \int_{\Theta} V^{\alpha}(p(\mu),\theta) \diff \mu(\theta), \label{eqn:valGen}
\end{align}
that specifies the expected $\alpha$-surplus of each market $\mu$.
Conversely, $\alpha$-IMG is equivalent to the no-information segmentation being a \emph{minimizer} of $V^{\alpha}$, or a maximizer of $-V^{\alpha}$, \emph{for all} priors.  This observation transform our problem to a class of optimization problems and allows us to give a unified proof for the two welfare-monotonicity properties.

\subsubsection{Step 2: Partial Inclusion and the First-Order Condition} \label{subsec:PI}
The optimization problem derived above nests another one in which the seller chooses an optimal price for each market.  We show that the partial-inclusion assumption allows us to replace the seller's optimal pricing problem with a first-order condition.

To understand the connection between partial inclusion and the first-order condition, consider some market $\mu$, and suppose for simplicity that it has binary support over two demands $\theta_L,\theta_H$.  Suppose there is full exclusion, i.e., $\overline{p}(\theta_L) < p(\theta_H)$, as shown in Panel (a) of \autoref{fig: no exclusion intuition}.  Even though each type has a revenue curve that is concave over its support, the revenue curve associated with market $\mu$ need not be concave, as shown in the figure.  As a result, the first-order condition does not pin down the optimal price.  We therefore proceed with two arguments.
\begin{figure}
	\centering
	\setlength{\tabcolsep}{10pt}
	\begin{tabular}{cc}
		
		\begin{tikzpicture}[ultra thick, scale=7]
			
			% Draw the x and y axes
			\draw[<->] (1,0) node[below]{price} -- (0,0) -- (0,.7) node[right]{revenue};
			
			\draw[domain=0:.3]    plot (\x,{20*\x*(0.3-\x)}) node[below,xshift=-4]{$\overline{p}(\theta_L)$};
			\draw[domain=0:.9]    plot (\x,{\x*(.9-\x)});
			
			\draw[dotted] (.45,.2) -- (.45,0) node[below,xshift=10]{$p(\theta_H)$};

            \draw[black!20,domain=0:.3]    plot (\x,{0.2*20*\x*(0.3-\x) + 0.8*\x*(.9-\x)});
			\draw[black!20,domain=0.3:.9]    plot (\x,{0.8*\x*(.9-\x)});

            \draw (.15,.45) node[above]{$R(p,\theta_L)$};

            \draw (.45,.2) node[above]{$R(p,\theta_H)$};

			\draw (.5,-.15) node[below]{(a)};

		\end{tikzpicture}&

		\begin{tikzpicture}[ultra thick, scale=7]
			
			% Draw the x and y axes
			\draw[dotted] (.3,.43) -- (.3,0) node[below,black,xshift=-3]{$p(\theta_L)$};			
			
			\draw[domain=0:.6]    plot (\x,{5*\x*(0.6-\x)}) node[below,black,xshift=5]{$\overline{p}(\theta_L)$};
			\draw[domain=0:.9]    plot (\x,{\x*(.9-\x)});

			\draw[dotted] (.45,.2) -- (.45,0) node[below,xshift=0,black]{$p(\theta_H)$};
			
			\draw[<->] (1,0) node[below]{price} -- (0,0) -- (0,0.7) node[right]{revenue};			
			
			\draw[black!20,domain=0:.6]    plot (\x,{0.3*5*\x*(0.6-\x) + 0.7*\x*(.9-\x)});
			\draw[black!20,domain=0.6:.9]    plot (\x,{0.7*\x*(.9-\x)});

			\draw (.5,-.15) node[below]{(b)};

		\end{tikzpicture}
	\end{tabular}
	\caption{(a) The gray curve, which represents the revenue curve of market $\mu$, is not concave.  (b) The optimal price for market $\mu$ is the price that satisfies the first-order condition over $[p(\theta_L),p(\theta_H)]$.}
	\label{fig: no exclusion intuition}
\end{figure}

First, we show that if partial inclusion is violated, then welfare-monotonicity does not hold, without relying on the first-order approach.  To do this, we consider the two cases of full exclusion, $\overline p(\theta_L)<p(\theta_H)$, and full inclusion, $p(\theta_L) < \underline p(\theta_H)$, separately and in each case construct an information structure that is good and one that is bad.  One of these four constructions is standard: If there is full exclusion, as is the case in the figure, then information has the potential to be beneficial by ``opening new markets''.
%\footnote{The classical literature that follows \citet{P20} implicitly relies on this observation to focus on the case where there is no full exclusion.  This observation also underlies results in the modern literature, \citet{BBM15} and \citet{Pra21}.}
The other three cases require new arguments that we discuss in detail in the proof.

Second, we show that under partial inclusion, the first-order approach is indeed valid.  The idea, as shown in Panel (b) of \autoref{fig: no exclusion intuition}, is that even though the revenue curve for a market is not necessarily concave, it is concave over the range $[\min_{\theta} p(\theta),\max_{\theta} p(\theta)]$ of lowest and highest optimal price of the types in the market, which are the only prices we need to consider to find the seller's profit-maximizing price.

\subsubsection{Step 3: Applying Duality}\label{sec: applying duality}
We next apply the duality framework of \citet{Kol18}, \citet{DM19}, \citet{DK24}, and \citet{KCW24}.  They key is to use those duality results to obtain optimality condition of the no-information segmentation for \emph{every} prior.

\begin{prop}\label{prop: duality converted to our setting}
Suppose there is partial inclusion.  Welfare-monotonicity holds if and only if there exists a function $\zeta$ such that for every ${p} \in I$ and every type $\theta$,
\begin{align}
	{p} \in \arg\max_{p'\in I} U(p',\theta) - \zeta({p},p')R_{p}(p',\theta),\label{eq: rewriting IMB IMG using duality}
\end{align}
where $U = V^{\alpha}$ characterizes $\alpha$-IMB and $U = - V^{\alpha}$ characterizes $\alpha$-IMG.
\end{prop}

Given the above result, the ``sufficiency'' direction of the proof is straightforward.  Suppose welfare-monotonicity holds for the binary family $\{\theta_L,\theta_H\}$, so by \autoref{prop: duality converted to our setting}, \cref{eq: rewriting IMB IMG using duality} holds for $\theta_L$ and $\theta_H$.  Suppose also that the demand curves of $\theta_L,\theta_H$ are bases that span some set $\mathcal{D}$ of demand curves, so every demand curve in the set is a linear combination of that of $\theta_L$ and $\theta_H$.  Then \cref{eq: rewriting IMB IMG using duality} holds for all demand curves in $\mathcal{D}$, and welfare-monotonicity holds for the entire set.

For the ``necessity'' part of the proof, again, one part of the argument is straightforward.  If welfare-monotonicity holds for $\mathcal{D}$, then \cref{eq: rewriting IMB IMG using duality} is satisfied for every type, and in particular it must hold for the two types $\theta_L,\theta_H$ that have the lowest and the highest optimal monopoly price in the family, so welfare-monotonicity must hold for the binary family that contains only the demand curves associated with $\theta_L,\theta_H$ as well.

To argue that the spanning condition of the theorem must be satisfied, we argue first that \emph{each} pair $\theta,\theta'$ of types pin down the function $\zeta$ over the interval $(p(\theta),p(\theta'))$ of prices that are in between their two optimal monopoly prices.  We then argue that if the spanning condition of the theorem is violated, the $\zeta$ functions that are pinned down with different pairs do not agree, and therefore a single function $\zeta$ does not exist.\footnote{The observation that the spanning condition is necessary for welfare-monotonicity is consistent with the finding in Corollary 1 of \citet{KCW24} that no-information is suboptimal for generic utility functions that the receiver might have in their Bayesian persuasion setting.  A difference is that the value functions of the receiver and the sender in their setting are unrelated whereas in our setting the value function $U$ and the seller's objective $R$ are both pinned down given the demand curves.  We use this property to give a complete characterization of our welfare-monotonicity properties.}  \cref{app: main result proof} provides detail of the proof. 

\subsection{Discussion and Connections to Classical and Modern Literature}\label{sec: discussion}
Having seen our characterization and the logic behind it, let us now compare our finding to those in the classical and modern literature on price discrimination.

\subsubsection{The Classical Literature} 
\paragraph{Comparison of conditions.} Our welfare-monotonicity conditions neither imply nor are implied by those in the classical literature that compare, in our language, only full information to no information.  On the one hand, our conditions are stronger because we characterize stronger properties.  On the other hand, the conditions in the literature are stronger because they provide only sufficient conditions.  For instance, as far as we are aware, existing results do not apply to constant-elasticity demand curves for which we show welfare-monotonicity.\footnote{In particular, this example violates the conditions of \citet{ACV10} that give the most general conditions for price discrimination to decrease total surplus.  Namely, their condition that the demand curve with a lower monopoly price has a lower curvature (is ``more concave'') than the demand curve with a higher monopoly price, is violated (it holds the opposite way) in this example.}  

\begin{example}[{\bf Constant-elasticity demand: $\frac{1}{2}$-IMB}]\label{example:CES} 
	Consider two demand curves $D(p,\theta_i) = (c+p)^{-\theta_i}$  for $i \in \{1,2\}$ and constant $c > 0$. Assume $1<\theta_2<\theta_1 \leq 2\theta_2 - 1$ and truncate the demand curves at $\frac{2c}{\theta_2-1}$.\footnote{These conditions ensure that \cref{ass1} holds.  The revenue curve associated with demand $(c+p)^{-\theta}$ is maximized at $\frac{c}{\theta-1}$, is concave below $\frac{2c}{\theta-1}$, and is convex above that threshold.  So by truncating the two demand curves at $\frac{2c}{\theta_2-1}$, both revenue curves are concave over the truncated interval.  The extra assumption $\theta_1 \leq 2\theta_2 - 1$ guarantees that optimal prices are in the interior.}
    Then information is monotonically bad for total surplus if and only if  $\theta_1 \leq\theta_2+\frac{1}{2}$.  Under this condition, information is monotonically bad for $\alpha$-surplus for any $\alpha \geq \frac{1}{2}$.
\end{example}

These demand curves have constant elasticities $\theta_1$ and $\theta_2$.\footnote{To be precise, these demand curves are a \emph{normalization} of constant-elasticity demands curves.  Starting from demand curve $p^{-\theta_i}$, which has constant elasticity $\theta_i$, and constant marginal cost $c>0$, if we normalize costs to zero by shifting prices by $c$, we obtain the demand functions used in this example.} The example shows that providing information to the seller monotonically decreases total  surplus exactly when the two demand elasticities are not too far apart. In other words, unless consumers have widely different elasticities with respect to price, price discrimination is bad for total (and consumer) surplus.

% Under these conditions, \cref{example:CES} shows that providing information to the monopolist monotonically decreases total surplus exactly when the two demand elasticities are not too far apart. In other words, unless consumers have very different elasticities with respect to price, price discrimination is bad for total (and consumer) surplus.

% Detailed calculations are in \cref{app:CES}.
% As far as we are aware, none of the existing results in the literature applies to this example.
% In particular, this example violates the conditions of Aguirre et al. (2010) that give the most
% general conditions for price discrimination to decrease total surplus.\footnote{Namely, their condition
% that the demand curve with a lower monopoly price has a \qm{lower} curvature (is “more concave”)
% than the demand curve with a higher monopoly price, is violated in this example. Our example
% therefore shows that this ranking of curvatures is not necessary for price discrimination to
% lower total surplus.}

\paragraph{Our three effects versus the two classical ones.} Although our conditions are different from those in the literature, they are in fact related in spirit. 
For this, recall the two effects from the classical literature discussed in the introduction: Price discrimination affects the total quantity produced, the ``output'' effect, and allows selling that total quantity at non-uniform prices, the ``misallocation'' effect.  Because the misallocation effect is a distortion that reduces total surplus, price discrimination decreases total surplus whenever the output effect is also negative. This logic therefore provides a sufficient condition for total surplus to fall, but cannot yield conditions for total surplus to rise because it offers no way to trade off the two effects when the output effect is positive (the misallocation effect is always negative).\footnote{Not all papers in the literature focus on this logic.  \citet{ACV10} give sufficient conditions for total surplus to increase or decrease with price discrimination using a different but related logic.}

Our result addresses precisely this limitation. It identifies three distinct effects of information and combines them into a single expression, giving a complete characterization of welfare-monotonicity.  Our effects are conceptually related to the classical ones but differ in important ways. Our price-curvature effect parallels the classical output effect, as both concern how price discrimination affects the seller’s aggregate decisions (average price versus total quantity). Similarly, our cross-groups price change effect parallels the classical misallocation effect, as both concern how price discrimination affects types non-uniformly. A first difference between our effects and the classical ones, however, is that all of our effects are measured in prices, not quantities.  This is important because we allow for general objectives, and not just total surplus.\footnote{Suppose we want to calculate consumer surplus in a market by taking the average of consumer surplus across all types in the market.  The quantity sold to the market is not a sufficient statistic for doing so because one has to calculate the price in order to determine how that quantity is divided across different types in the market.}   A second difference is that our within-groups price change effect has no counter-part and arises because we study partial segmentations: Even consumers of one type are treated differently by price discrimination.

\paragraph{Flexible welfare measures.} The flexibility of our characterization, which allows for arbitrary weights on consumer and producer surplus, helps paint a more comprehensive picture of the effects of price discrimination than fixing any particular objective would. To elaborate, consider the following corollary of our result. 

\begin{corol} \label{prop:simplifiedFmility}
  Consider the family of demand curves
  \begin{align*}
      \mathcal{D} = \Big\{a(\theta)D(p) + b(\theta)\Big\}_{\theta}
  \end{align*}
  for some demand curve $D$.  Then $\alpha$-IMB ($\alpha$-IMG) holds if and only if
  \begin{align}
        (2\alpha-1) p + \alpha (\frac{pD'(p)}{R''(p)})\label{eq: simplified expression for aD+b}
  \end{align}
  is increasing (decreasing) over $[\min\limits_{\theta}p(\theta),\max\limits_{\theta}p(\theta)]$ and there is partial inclusion.\footnote{We denote the first and second derivative of $R(p)$ by $R'(p)$ and $R''(p)$, and similarly $D'(p)$ is the derivative of $D(p)$.}
\end{corol}

\autoref{prop:simplifiedFmility} implies that information is monotonically bad for total surplus if and only if $\frac{pD'(p)}{R''(p)}$ is increasing.  This finding is consistent with those of \citet{Cow07} and \citet{ACV10} who show that if this expression is increasing, then total surplus is lower under full segmentation than under uniform pricing. Additionally, our result shows that when the expression is decreasing, then information is monotonically good for total surplus.  In this case, does total surplus increase because consumer surplus also increases, or because any loss in consumer surplus is offset by a larger gain in producer surplus?  The corollary clarifies: Consumer surplus, and therefore total surplus, increases when \cref{eq: simplified expression for aD+b} is decreasing for $\alpha = 1$.  Otherwise, consumer surplus may decrease, but any losses are offset by an increase in producer surplus when \cref{eq: simplified expression for aD+b} is decreasing for $\alpha = \frac{1}{2}$.  In Appendix~\ref{onlineapp: three cases}, we give an example demonstrating each possibility based on the curvature of a parametric demand function, qualitatively shown in \autoref{fig: aD+b example summary IMGIMB}.
\begin{figure}
	\centering	
	\begin{tikzpicture}[ultra thick, scale=2.5]
		\draw[->] (-3.5,0) -- (2.5,0) node[below,xshift=-15]{curvature};

		\draw (0.53,.08) -- (0.53,-.08) node[below]{high};
		\draw (-1.53,.08) -- (-1.53,-.08) node[below]{low};
		
		\draw[dashed] (0.53,0) -- (0.53,.75);
		\draw[dashed] (-1.53,0) -- (-1.53,.75);
		
		%\draw (1,.7) node{$\hat{\alpha} \in (0,\frac{1}{2})$};
		\draw (1.5,.4) node{Bad for both TS and CS};

		%\draw (-2,.7) node{$\hat{\alpha} \in (\frac{1}{2},1)$};
		\draw (-2.65,.4) node{Good for TS but bad for CS};

		%    \draw (-.5,.7) node{$\hat\alpha=1$};
		\draw (-.5,.4) node{Good for both TS and CS};
	\end{tikzpicture}
	\caption{The three possible cases for how information affects total surplus (TS) and consumer surplus (CS) in a parametric example.}
	\label{fig: aD+b example summary IMGIMB}
\end{figure}

\subsubsection{Approaching Step Functions} 
Having seen the connection between our conditions and those in the classical price discrimination literature, let us compare our findings to those in \citet{BBM15}. Because the demand curves in our setting must be downward-sloping over a non-empty interval, they cannot be step functions, which are considered in \citet{BBM15}.  One can nonetheless ask what happens when we \emph{approach} step functions within our framework.  Because we allow different types to have different supports of prices, it is in fact possible to approach each step function with a strictly decreasing demand curve that has a strictly concave revenue curve over a vanishingly small interval.  But as these intervals get small, the partial inclusion condition will eventually be violated, in which case our result says that welfare-monotonicity must be violated, consistent with what \citet{BBM15} find.\footnote{Even though the findings of \citet{BBM15} are not stated in terms of welfare-monotonicity, they imply that both welfare-monotonicity properties are violated.  In particular, with two types representing step functions where $p(\theta_L) < p(\theta_H)$, any segmentation in which some market has price $p(\theta_H)$ can be refined in a way that increases consumer surplus and therefore weighted surplus.   And any segmentation in which some market has price $p(\theta_L)$ can be refined in a way that reduces consumer surplus while keeping producer surplus fixed, reducing weighted surplus for any weight.} {The next Corollary formalizes this observation.}

\begin{corol}
    Consider a family of demand curves $\mathcal{D}^{\epsilon}$ that uniformly converges to a family of unit-demand curves as $\epsilon$ goes to zero while satisfying our \autoref{ass1} for every $\epsilon > 0$.  Then, for small enough $\epsilon$, the partial-inclusion condition is violated and therefore information is neither monotonically good nor bad.
\end{corol}

\subsubsection{Beneficial Information without Opening New Markets} 
Another insight from our result, that relates both to the classical and the modern literature, is that beneficial information generally exists even without opening new markets.  

A folklore wisdom in the classical price discrimination literature is that information has the potential to be beneficial if it opens new markets, which happens if there is exclusion, $\overline{p}(\theta_L) < p(\theta_H)$, because then separating the demand curves might prompt the seller to lower the price for type $\theta_L$ consumers without increasing it for type $\theta_H$ consumers.  

Following this observation, the modern and the classical literature then take different directions.  \citet{BBM15} focus on this force, asking what can happen when information is used to open new markets and demonstrating its power.\footnote{In their setting, simply separating the types does not increase consumer surplus, so more careful constructions are needed.  The main force, however, is the ability to create segmentations that decrease prices for some consumers without increasing it for others.}  The classical literature imposes a non-exclusion assumption in order to examine whether, without opening new markets, price discrimination can improve welfare.  Our result shows that it generally can.  Indeed, when there is partial inclusion (so no new markets are opened) but as long as one of the two other conditions of the characterization is violated (for example, the spanning condition, which is a demanding one), information does not affect welfare monotonically, so there is some information that is beneficial.  We pursue this case further in the next section.

\section{Beyond Monotonicity: Bounds and Improvements} \label{sec:gen}
\autoref{thm:NSC} immediately leads to two followup questions:  In cases where information does not affect welfare monotonically, what are the bounds on the effects, and what are the best and worst ways to provide additional information?  Our second main result answers these two questions.  

To simplify the analysis we assume that there is a finite set of types, where $n$ denotes the number of types.  In parallel to our analysis of the binary case where we represented a market by a single scalar, we represent a market by a $n-1$-dimensional vector $\mu = (\mu_2,\ldots,\mu_{n})$, where $\mu_i$ is the probability of each type $\theta_i$ for $i \neq 1$ and $1 - \sum_{i=2}^{n} \mu_i$ is the probability of type 1.

The first part of the result provides tight bounds on the marginal value of information.  To formalize, for any two segmentation $\sigma, \sigma'$ where $\sigma$ is a mean-preserving spread of $\sigma'$, define 
\begin{equation}
\Delta V^{\alpha}(\sigma,\sigma'):=\frac{V^{\alpha}\left(\sigma\right)-V^{\alpha}\left(\sigma'\right)}{\mathbb{E}_{\sigma}\left[\left\Vert \mu\right\Vert _{2}^{2}\right]-\mathbb{E}_{\sigma'}\left[\left\Vert \mu\right\Vert _{2}^{2}\right]},\label{eq: change in value rate}
\end{equation}
where $||\cdot||$ is the Euclidean norm. The numerator is the difference between the weighted surplus of the two segmentations.  The denominator is a measure of the additional information contained in $\sigma$ compared to $\sigma'$.  As a result, $\Delta V^{\alpha}(\sigma,\sigma')$ measures the \emph{rate} of surplus change with respect to the amount of additional information.  The theorem identifies the lowest and the highest values that $\Delta V(\sigma,\sigma')$ can take over all pairs of segmentation $\sigma \underset{\text{MPS}}\succeq \sigma'$.  This result is therefore a \emph{quantitative} counterpart to our (qualitative) first result, which is concerned with identifying when the bounds are both (weakly) positive ($\alpha$-IMG), both (weakly) negative ($\alpha$-IMB), or have different signs (the non-montone case).

The second part of the result identifies the best and the worst ways to provide additional information.  We particularly focus on the case where a small amount of additional information is provided, and ask what the best and worst ``direction'' for providing information is.  To formalize this, for any segmentation $\sigma$ and $\epsilon \leq 1$, let $\sigma_{\epsilon}$ be 
a segmentation that moves $\sigma$ towards the prior market by replacing each market $\mu$ in the support of $\sigma$ with $\epsilon \mu + (1-\epsilon) \mu_0$.\footnote{Formally the support of $\sigma_{\epsilon}$ is $\{\epsilon \mu + (1-\epsilon) \mu_0\  |\  \mu \in supp(\sigma)\}$, and $\sigma_{\epsilon}(M) = \sigma(\{\frac{\mu - (1-\epsilon)\mu_0}{\epsilon} \ | \ \mu \in M \})$ for every set of markets $M$ in the support of $\sigma_{\epsilon}$.}  When $\epsilon = 1$, the segmentation is equal to $\sigma$, and as $\epsilon$ decreases, more weight is put on the prior market, so less information is provided in some formal sense.   We say that a segmentation $\sigma$ is marginally optimal if for every segmentation $\sigma'$ of the same size, i.e., $\mathbb{E}_{\sigma}\left[\left\Vert \mu\right\Vert _{2}^{2}\right] = \mathbb{E}_{\sigma'}\left[\left\Vert \mu\right\Vert _{2}^{2}\right]$, we have $V^{\alpha}(\sigma_{\epsilon}) \geq V^{\alpha}(\sigma'_{\epsilon})$ for small enough $\epsilon$, and is marginally pessimal if the opposite holds, $V^{\alpha}(\sigma_{\epsilon}) \leq V^{\alpha}(\sigma'_{\epsilon})$.
%The (normalized) marginal value of segmentation $\sigma$ is defined to be
%\begin{align*}
%    MV^{\alpha}(\sigma) = \lim_{\epsilon \rightarrow 0} \frac{V^{\alpha}(\sigma_{\epsilon}) - W^{\alpha}(\mu_0)}{\epsilon^2(\mathbb{E}_{\sigma}\left[\left\Vert \mu\right\Vert _{2}^{2}\right] - \left\Vert \mu_0\right\Vert _{2}^{2})}.
%\end{align*}
%This definition captures the marginal value of moving towards segmentation $\sigma$ by a small amount, starting from the prior market.  The value is normalized by the magnitude of information, and allows us to compare different segmentations with different magnitudes.  We say that a segmentation $\sigma$ is marginally optimal if $MV^{\alpha}(\sigma) \geq MV^{\alpha}(\sigma')$ for any $\sigma'$ and is marginally pessimal if $MV^{\alpha}(\sigma) \leq MV^{\alpha}(\sigma')$ for any $\sigma'$.\footnote{An implication of or normalized definition of value is that if a segmentation $\sigma$ is marginally optimal, then it maximizes $V^{\alpha}_{\epsilon}(\sigma,0)$ among all segmentations of the same magnitude.}  
For a given vector $\nu \in \mathbb{R}^{n-1}$, we say that a segmentation spreads the prior market in the direction $\nu$ if for every market $\mu$ in the segmentation, $\mu - \mu_0$ is proportional to $\nu$.  The result shows that marginally optimal (pessimal) segmentations are those that spreads the prior in a particular direction and identifies that direction.

For both the bounds and the best and worst directions, an important role is played by the lowest and the highest eigenvalues and eigenvectors of the Hessian matrix $\nabla^{2}W^{\alpha}(\mu) = (\partial_{\mu_i}\partial_{\mu_j} W^{\alpha}(\mu))_{i,j=2,\ldots,n}$ of the value function $W^{\alpha}(\mu)$ (defined in \cref{eqn:valGen}).
The largest possible benefit of providing information is proportional to the largest eigenvalue, and the best direction to provide information is in the direction of the associated eigenvector, which we simply call the ``largest" eigenvector (and analogously for the worst benefit and direction).  

These eigenvalues and eigenvectors exist and can be specified in closed form.  Indeed, they directly reflect and generalize the three effects of information from our two-type analysis. Because exploring the connection and giving the expressions explicitly requires some additional notation, we defer doing so to after the statement of the result, and in the statement only note that these values can be identified in closed form.

\begin{thm}[\bf Bounds and Improvements]
\label{thm:bounds} Suppose $\mathcal{D}$ contains a finite number of demand curves and there is partial inclusion.
\begin{enumerate}
    \item[(i)] For any market $\mu$, the Hessian $\nabla^{2}W^{\alpha}(\mu)$ has at most two non-zero eigenvalues. The lowest and the highest eigenvalues of the Hessian, $\underline{\lambda}^{\alpha}\left(\mu\right)$ and $\overline{\lambda}^{\alpha}\left(\mu\right)$, and their eigenvectors $\underline{v}^{\alpha}(\mu)$ and $\overline{v}^{\alpha}(\mu)$, exist and can be expressed in closed form.  The two eigenvalues have weakly opposite signs, $\underline{\lambda}^{\alpha}\left(\mu\right) \leq 0 \leq \overline{\lambda}^{\alpha}\left(\mu\right)$.
    \item[(ii)] The lower and upper bounds on the {marginal value of} information are proportional to the lowest and the highest eigenvalues,
    \begin{align}
\inf_{\sigma \underset{\text{MPS}}\succeq \sigma'}\Delta V^{\alpha}(\sigma,\sigma') = \frac{1}{2}\min_{\mu\in\Delta\Theta}\underline{\lambda}^{\alpha}\left(\mu\right),\\
\sup_{\sigma \underset{\text{MPS}}\succeq \sigma'}\Delta V^{\alpha}(\sigma,\sigma') = \frac{1}{2}\max_{\mu\in\Delta\Theta}\overline{\lambda}^{\alpha}\left(\mu\right).\label{eq: bounds}
\end{align}
\item[(iii)] A segmentation is marginally optimal (pessimal) if and only if it spreads the prior market $\mu_0$ in the direction of the largest (smallest) eigenvector $\overline{v}^{\alpha}(\mu_0)$.
\end{enumerate}
\end{thm}

% Let us point out that it applies more generally to the setting of \citet{KCW24} in which an agent who receives information (the seller in our setting) chooses a one-dimensional action (the price) that is identified by a first-order condition.

We prove \autoref{thm:bounds} by using the Taylor expansion of the difference between surplus at $\sigma$ and $\sigma'$, with random markets $\mu$ and $\mu'$, to approximate it in a quadratic form
\begin{align*}
    \Delta:= V^{\alpha}\left(\sigma\right)-V^{\alpha}\left(\sigma'\right) \approx \frac{1}{2} \mathbb{E}\left[\left(\mu'-\mu\right)^{T}\nabla^{2}W^{\alpha}\left(\tilde{\mu}\right)\left(\mu'-\mu\right)\right],
\end{align*}
where $\tilde{\mu}$ is some market that lies on the line that connects $\mu$ and $\mu'$.  We then apply the Courant–Fischer–Weyl min-max principle to bound the above expression using eigenvalues.\footnote{See Theorem 1.3.2 in  \citet{Tao23}.}  In particular, we show
\begin{align*}
    \Delta \leq \frac{1}{2} \mathbb{E}\left[\left(\mu'-\mu\right)^{T}  \left(\mu'-\mu\right) \overline{\lambda}^{\alpha}\left(\tilde\mu\right) \right] \leq \frac{1}{2} \left( \mathbb{E}_{\sigma}\left[\left\Vert \mu\right\Vert _{2}^{2}\right]-\mathbb{E}_{\sigma'}\left[\left\Vert \mu\right\Vert _{2}^{2}\right] \right)\max_{\mu\in\Delta\Theta}\overline{\lambda}^{\alpha}\left(\mu\right).
\end{align*}
A similar application of the Taylor expansion, applied to small perturbations of the prior market, gives the marginally optimal and pessimal directions.

The result has a simple intuition.  Recall that information is monotonically bad if the value function $W^{\alpha}$ is concave, and is monotonically good if it is convex.  So it is natural that how bad or good information is depends on how concave or convex the value function is.  Eigenvalues measure curvature along their eigenvectors:   The largest eigenvector identifies a direction in which the value function is ``most convex'', and the eigenvalue measures how convex it is in that direction.  The best information to provide is along that most convex direction, and the value of providing information in that direction is captured by the eigenvalue, the degree of convexity.  The extreme cases of welfare-monotonicity are when the value function is concave or convex, in which case both the eigenvalues of the Hessian are negative or positive, respectively.

Because \autoref{thm:bounds} is a quantitative result, it allows us to formalize the statement that if a family of demand curves approximately but not exactly satisfies the conditions of \autoref{thm:NSC}, then welfare-monotonicity approximately holds.  This is simply because the bounds of \autoref{thm:bounds}, and in particular the eigenvalues of the matrix $\nabla^{2}W^{\alpha}(\mu)$, change continuously as we perturb a family of demand curves, so if one of the bounds is zero for a family of demand curves, then the corresponding bound approaches zero as we approach that family of demand curves.
%\footnote{Formally, consider a finite family of demand curves $\mathcal{D}^{\epsilon} = \{D_1^\epsilon,\ldots,D_n^{\epsilon}\}$ such that as $\epsilon$ goes to zero, each $D_i^\epsilon$ converges to $D_i$ and the family of demand curves $\mathcal{D} = \{D_1,\ldots,D_n\}$ satisfies welfare-monotonicity. Because the limiting family satisfies surplus-montonicity, either the lower bound or the upper bound in \autoref{thm:bounds} is zero.  Then, because of continuity, the corresponding bound for the family $\mathcal{D}^{\epsilon}$ also approaches zero, meaning welfare-monotonicity approximately holds.  %Indeed, \autoref{thm:bounds} is formally a stronger result than \autoref{thm:NSC} because the former can be proved to use the latter.  Nonetheless, proving \autoref{thm:NSC} using \autoref{thm:bounds} requires some mechanical derivations and is not simpler than proving it via concavification and duality the way we have done in \autoref{sec: two demands} and \autoref{sec:method}, which we find to be more insightful.}

\autoref{thm:bounds} can also be used to provide bounds on the magnitude of the change in weighted surplus $V^{\alpha}\left(\sigma\right)-V^{\alpha}\left(\sigma'\right)$, as opposed to the \emph{rate} of surplus change as measured in \cref{eq: change in value rate}.  Indeed, for any two segmentations $\sigma,\sigma' \in S(\mu_0)$ where $\sigma$ is a mean-preserving spread of $\sigma'$, we can bound the amount of additional information that $\sigma$ contains compared to $\sigma'$ as follows,
\begin{align*}
    \mathbb{E}_{\sigma}\left[\left\Vert \mu\right\Vert _{2}^{2}\right]-\mathbb{E}_{\sigma'}\left[\left\Vert \mu\right\Vert _{2}^{2}\right] \leq 1 - \left\Vert \mu_0\right\Vert _{2}^{2},
\end{align*}
because the amount of additional information is highest when $\sigma$ is the full-information segmentation and $\sigma'$ is the no-information segmentation.  As a result, substituting the above bound into that of \autoref{thm:bounds}, we get the following bounds on the magnitude of the change in surplus,
\begin{equation}
\frac{1 - \left\Vert \mu_0\right\Vert _{2}^{2}}{2}\min_{\mu\in\Delta\Theta}\underline{\lambda}^{\alpha}\left(\mu\right)\leq V^{\alpha}\left(\sigma\right)-V^{\alpha}\left(\sigma'\right)\leq\frac{1 - \left\Vert \mu_0\right\Vert _{2}^{2}}{2}\max_{\mu\in\Delta\Theta}\overline{\lambda}^{\alpha}\left(\mu\right).\label{eq: bounds nonratio}
\end{equation}
Whereas the bound in \cref{eq: bounds nonratio} measures how much surplus increases or decreases, the bound in \cref{eq: bounds} is on the rate of change in surplus.  As a result, \cref{eq: bounds nonratio} is useful in situations where we know nothing beyond the fact that the seller is collecting additional information, whereas \cref{eq: bounds} gives tighter bounds when we also know (or can bound) \emph{how much} additional information the seller is collecting. {In particular, the bounds in \cref{eq: bounds nonratio} are loose when the amount of additional information is small, because in such cases $1 - \left\Vert \mu_0\right\Vert _{2}^{2}$ is a loose bound on the amount of additional information that $\sigma$ contains compared to $\sigma'$.}

Even though the definition of what a marginally optimal segmentation is seems to start from the prior distribution, corresponding to no information, it can also be used to identify how to improve any segmentation.  That is, suppose we start from some arbitrary segmentation, and want to optimally provide a small amount of additional information.  The additional information, that spreads each posterior $\mu$ in the support of the segmentation in a mean-preserving way, must do so in the direction that is marginally optimal for that posterior, which is in the direction of the largest eigenvector $\overline{\lambda}^{\alpha}(\mu)$.

Finally, let us point out that appropriate generalization of the theorem applies to Bayesian persuasion problems in which the receiver's action is specified by a first-order condition (for example, \citealp{KCW24}).

%generalization of the theorem applies to the setting of \citet{KCW24} in which an agent who receives information (the seller in our setting) chooses a one-dimensional action (the price) that is identified by a first-order condition.

%This intuition more generally applies to Bayesian persuasion problems in which the receiver's action is specified by a first-order condition (for example, \citealp{KCW24})

\subsection{Closed-form Expressions and the Three Effects of Information}\label{sec: general closed form}
As we mentioned above, the lowest and highest eigenvalues and eigenvectors of the Hessian matrix of the value function, which represent the magnitude and directions of best information, can be expressed in closed form.  We provide these closed-form expressions here. An examination of these expressions reveals that they generalize the three effects of information we discussed in our two-type analysis.

Before discussing the eigenvalues and the eigenvectors of the Hessian of $W^{\alpha}$, let us first examine the Hessian itself. This Hessian can be written as a sum of three terms.  In particular, each element $W^{\alpha}_{i,j}(\mu)$ of the Hessian matrix, which is a cross-partial derivatives of the value function with respect to $\mu_i,\mu_j$ can be written as
\begin{align}
    W^{\alpha}_{i,j}(\mu) & =  p_i(\mu) p_j(\mu) \mathbb{E}\big[V^{\alpha}_{pp}(\theta,p(\mu))\big] \nonumber \\
    & + p_j(\mu) \Big(V^{\alpha}_p(\theta_i,p(\mu)) - V^{\alpha}_p(\theta_1,p(\mu))\Big)  + p_i(\mu) \Big(V^{\alpha}_p(\theta_j,p(\mu)) - V^{\alpha}_p(\theta_1,p(\mu))\Big)\nonumber\\
    & + p_{i,j}(\mu) \mathbb{E}\big[V^{\alpha}_p(\theta,p(\mu))\big].  \label{eq: partial derivarives of W}
\end{align}
This cross-partial derivative {is a generalization of \cref{eq: the three effects} and} represents the three effects of information.  The first term represents the within-group price change effect, the second one the cross-groups price change effect, and the third one the price curvature effect.

The key observation that allows us to obtain closed-form expressions for the eigenvalues and eigenvectors of the Hessian is that the Hessian itself is a low rank matrix with a separable form.  For this, let us first write the Hessian compactly in matrix form using \cref{eq: partial derivarives of W},
\begin{align*}
    \nabla^2 W^{\alpha}(\mu) &=   \nabla p(\mu)^T \nabla p(\mu) \mathbb{E}\left[V^{\alpha}_{pp}(\theta,p(\mu))\right] \\
    & + (\Delta V^{\alpha}_p(\mu))^T \nabla p(\mu) + \nabla p(\mu)^T \Delta V^{\alpha}_p(\mu)\\
    & +\nabla^2 p(\mu) \mathbb{E}\left[V^{\alpha}_p(\theta,p(\mu))\right],
\end{align*}
where $\nabla p(\mu) = (p_2(\mu),\ldots,p_n(\mu))$ is the gradient vector of optimal price, $\nabla^2 p(\mu)$ is the Hessian of the optimal price, the vector 
\begin{align*}\Delta V^{\alpha}_p(\mu) = \left(V^{\alpha}_p(\theta_i,p(\mu)) - V^{\alpha}_p(\theta_1,p(\mu))\right)_{i=2,\ldots,n}\end{align*} 
is one that compares the derivative of the value function at every type $i \neq 1$ to that of type 1, and the superscript $T$ transposes each row vector to a column vector.   

Notice that the {first to third} term on the right hand side of the Hessian are already separated as products of a column vector and a row vector (multiplied possibly by a scalar).  Only the {last} term, involving the Hessian $\nabla^2 p(\mu)$ of the optimal price, is not already separated as products of two vectors.   We show that this term can also be separated, which allows us to write
\begin{align}
    \nabla^2 W^{\alpha}(\mu) = x(\mu)^T \nabla p(\mu) + \nabla p(\mu)^T x(\mu),\label{eq: hessian separated}
\end{align}
for some $x(\mu)$ that is defined as follows,
\begin{align*}
    x(\mu) =  \frac{1}{2} \mathbb{E}[V^{\alpha}_{pp}(\theta,p(\mu))] \nabla p(\mu) +\Delta V^{\alpha}_p(\mu) - \frac{\mathbb{E}[V^{\alpha}_{p}(\theta,p(\mu))]}{\mathbb{E}[R_{pp}(\theta,p(\mu))]} \bigg(\Delta R_{pp}(\mu) + \frac{1}{2}\mathbb{E}[R_{ppp}(\theta,p(\mu))] \nabla p(\mu)\bigg),
\end{align*}
where the vector 
\begin{align*}\Delta R_{pp}(\mu) = \Big(R_{pp}(\theta_i,p(\mu)) - R_{pp}(\theta_1,p(\mu))\Big)_{i=2,\ldots,n}\end{align*} 
is one that compares the second derivative of the revenue function at every type $i \neq 1$ to that of type 1. 

The vector $x$ is a direct generalization of the expression in \cref{eq:main} that characterizes welfare-monotonicity for binary types.  Recall that welfare-monotonicity is equivalent to the monotonicity of \cref{eq:main} as a function of price.  If we take the derivative of \cref{eq:main}, and evaluate it at the optimal price of a market, we obtain the equivalent of vector $x$  in the binary case.

Given the separable form of the Hessian, its two eigenvalues and eigenvectors can be written in closed form using the same two vectors $\nabla p$ and $x$,
\begin{alignat*}{2}
\overline{\lambda}\left(\mu\right)= & \nabla p\left(\mu\right) x\left(\mu\right)^T &&+\left\Vert \nabla p\left(\mu\right)\right\Vert _{2}\left\Vert x\left(\mu\right)\right\Vert _{2}\\
\underline{\lambda}\left(\mu\right)= & \nabla p\left(\mu\right) x\left(\mu\right)^{T}&&-\left\Vert \nabla p\left(\mu\right)\right\Vert _{2}\left\Vert  x\left(\mu\right)\right\Vert _{2}\\
\overline{v}\left(\mu\right)= & \nabla p\left(\mu\right)\left\Vert x\left(\mu\right)\right\Vert _{2}&&+\left\Vert \nabla p\left(\mu\right)\right\Vert _{2}     x\left(\mu\right)\\
\underline{v}\left(\mu\right)= & \nabla p\left(\mu\right)\left\Vert x\left(\mu\right)\right\Vert _{2}&&-\left\Vert \nabla p\left(\mu\right)\right\Vert _{2}     x\left(\mu\right).
\end{alignat*}
%To understand what these equations say, consider for example the first one, which defines the largest eigenvalue.  Recall that $x$ is the vector of pairwise curvatures normalized by marginal changes in price.  When we multiply $x$ by the gradient of the price, the derivatives of the price across different dimensions cancel and we simply get the summation of pairwise curvatures across all dimensions.  But simply summing up these pairwise curvatures is not enough, and we have to account for the sizes of the two vectors, which is what the second term does.

\subsection{Two Examples}
To conclude this section, let us apply \autoref{thm:bounds} to two examples.  We will show that even though the first example does not satisfy welfare-monotonicity, we can think of it as an instance where information is essentially bad because the potential benefit is much smaller than the potential harm.  To identify best directions in an example where the possible gains are potentially large, we study a second example.

Recall from \autoref{example:CES} that information is monotonically bad for total surplus for the binary family of constant-elasticity demands $\{(1+p)^{-\theta_1},(1+p)^{-\theta_2}\}$ when $1 < \theta_2 < \theta_1 \leq \theta_2 + \frac{1}{2}$, for example when $\theta_1 = 2$ and $\theta_2 = \frac{3}{2}$.  However, information is \emph{not} monotonically bad for total surplus when we add a third constant-elasticity demand curve, such as when the family is $\left\{ \left(1+p\right)^{-3/2},\left(1+p\right)^{-\theta},\left(1+p\right)^{-2}\right\}$,  $\frac{3}{2} < \theta < 2$, because the spanning condition of \autoref{thm:NSC} is violated for this family.  Let us now apply \autoref{thm:bounds} to this family to obtain tight welfare bounds.

\autoref{tab:Highest-and-lowest} contains
the results of the numerical calculations of the lower and the upper bounds of \autoref{thm:bounds} associated with different values of $\theta$, for the case of total surplus. 
As can be seen in the table, the upper bound on the marginal value of information is at most in the order of magnitude
of $10^{-4}$, whereas the lower bound is much larger (in absolute value), between $-0.28$ and $-0.46$.  Recall also that these bounds are tight in the sense formalized in \autoref{thm:bounds}. We interpret these observations as saying that while information is not monotonically bad for this family, the possible benefit from segmentation is much smaller compared the to potential harm, and information is essentially monotonically bad for total surplus (and therefore also for consumer surplus).
\begin{table}[!h]
	\def\arraystretch{1.5}
\begin{centering}
\begin{tabular}{|c|c|c|}
\hline 
$\theta$ & Lower bound &  Upper bound \tabularnewline
\hline 
\hline 
1.9 & $-0.460$ & $4.6\times10^{-5}$\tabularnewline
\hline 
1.8 & $-0.395$ & $1.47\times10^{-4}$\tabularnewline
\hline 
1.7 & $-0.332$ & $2.19\times10^{-4}$ \tabularnewline
\hline 
1.6 & $-0.282$ & $1.48\times10^{-4}$\tabularnewline
\hline 
\end{tabular}
\par\end{centering}
\caption{The lower and upper bounds of \autoref{thm:bounds} for a parametric family of three constant-elasticity demand curves.}
\label{tab:Highest-and-lowest}
\end{table}

In order to illustrate the second part of the theorem that identifies best directions, let us move on to another example in which the possible benefits of information are significant relative to its potential harms.  \autoref{fig: 1-p best direction} shows the results of numerical calculations for the family of demand curves $\{1-p^\theta\}_{\theta= 0.01,0.3,0.9}$, for the case of consumer surplus.  
\begin{figure}[tbph]
	\centering
	\setlength{\tabcolsep}{2pt}
	\begin{tabular}{cc}
		
		\begin{tikzpicture}[ultra thick, scale=5]
			
				\node {\includegraphics[width=8cm]{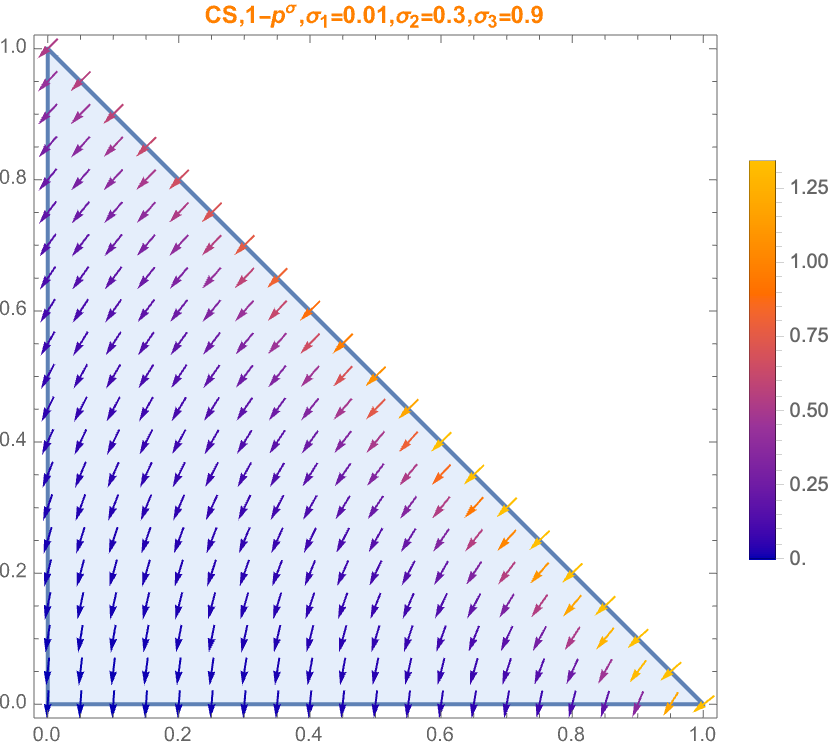}};

			\draw (0,-.75) node[below]{(a)};

		\end{tikzpicture}&

		\begin{tikzpicture}[ultra thick, scale=5]
			
				\node {\includegraphics[width=8cm]{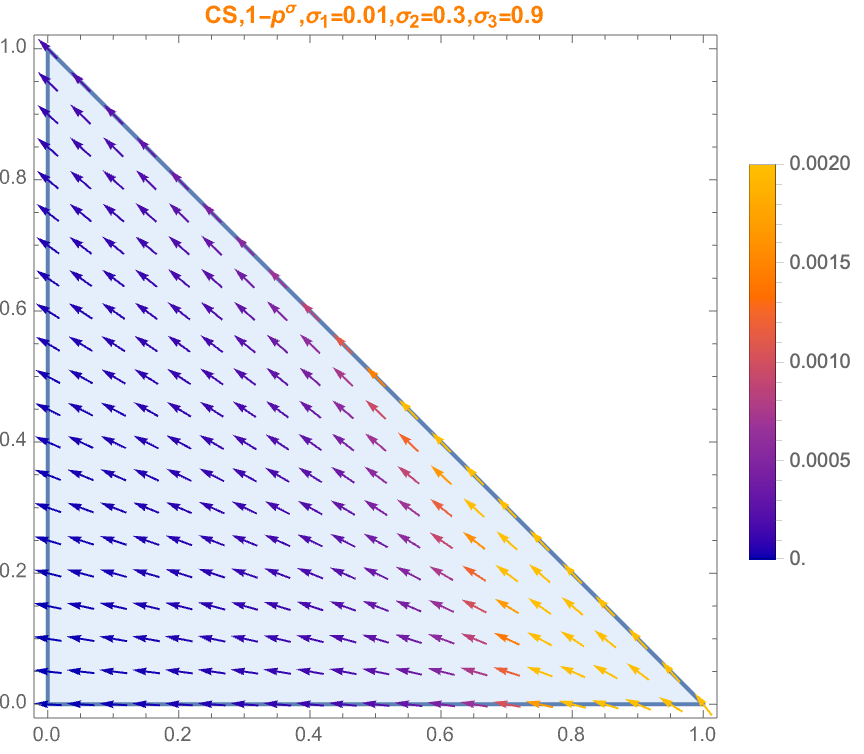}};

			\draw (0,-.75) node[below]{(b)};

		\end{tikzpicture}
	\end{tabular}
	\caption{The (a) best and (b) worst directions for information.}
	\label{fig: 1-p best direction}
\end{figure}

Each triangle in the figure shows the set of all markets, with the market containing only the ``low'' type (the one with $\theta = 0.01$) as the right corner, the one with only the ``middle'' type ($\theta = 0.3$) as the top corner, and the one with only the ``high'' type ($\theta = 0.9$) as the bottom left corner.  For each market in Panel (a), the arrow shows the direction of best information, that is, how the market should be split to maximize the gains from segmentation, and the color shows the magnitude of improvement (from blue for lower to yellow for higher). Similarly, the arrows in Panel (b) represent the worst direction and its magnitude.   Comparing the scales to the right of the two panels shows that the potential gains are significant relative to potential harms for consumer surplus (1.25 compared to 0.002), and therefore also for total surplus.  So let us focus on Panel (a) and discuss the best directions.

Roughly speaking, for markets towards the northeast, that is, those in which the probability of the high type is small, information has a large benefit and the best direction is diagonal.  Such a segmentation keeps the difference between the probabilities of the low and the middle types constant across segments. 
That is, it splits the prior market to two identical segments, and then takes an equal measure of the low and middle type consumers from one segment and moves them to the other one.  
For markets towards the southwest, that is, those in which the probability of the high type is large, information has a small benefit and the best direction is vertical.  Such a segmentation keeps the probability of the low type fixed, providing no information about the low type, but separates the middle type from the high type.

\section{Conclusions} \label{sec:conc}
We develop a model of price discrimination based on information structures that connects the classical literature to the modern one.  Like the classical literature, the primitive of our model is a set of demand curves and a distribution over them.  Like the modern literature, we study the class of all information structures.

Our first main result gives conditions under which welfare increases monotonically, decreases monotonically, or changes non-monotonically when segmentations are refined.  A main insight of the characterization is that information affects welfare in three ways: the within-group price change effect, the cross-groups price change effect, and the price curvature effect.  While two of these effects are related in spirit to the output and misallocation effects identified in the classical literature (but with a main difference that they are concerned with price, not quantity), the third one, the within-group price change effect, has no counter-part and arises because even consumers of a given type are treated non-uniformly by partial information.  Another insight of the first result is that information generally has the potential to increase welfare even without opening new markets.

Our second main result studies the non-monotone case further. It provides bounds on the effects of information and identifies the best direction to provide additional information.  The result is based on a simple yet novel insight that eigenvalues of the Hessian matrix of the value function measure how good and bad information can be in the directions specified by their corresponding eigenvectors.  Notably, in our setting these eigenvalues (and their eigenvectors) can be identified in closed form, reflecting and generalizing the three effects of information from our two-type analysis.

\bibliographystyle{aer}
\bibliography{FHSbib}

\begin{appendix}
\begin{center}
{\Huge \bf Appendix}
\end{center}

\section{Proofs}

\subsection{Proof of \autoref{prop: duality converted to our setting}}\label{app: proof of strong duality starting point}
First notice that the optimal price $p(\mu)$ for any market $\mu$ must be at least $\min_{\theta}{p}(\theta)$ and at most $\max_{\theta}{p}(\theta)$.  Because of partial inclusion, all revenue curves are strictly concave in this range.  
	As a result, the optimal price $p(\mu) \in [\min_{\theta}{p}(\theta), \max_{\theta}{p}(\theta)]$ is unique and is characterized by
	\begin{align*}
		\int_{\Theta} R_p(p(x),\theta) \diff \mu (\theta) = 0.
	\end{align*}

 It is convenient to reformulate the problem as one of choosing a joint distribution $G$ over types and prices subject to the constraint that the marginal of $G$ over the types agrees with $\mu$ and the obedience constraint that conditioned on a given price recommendation, the expected marginal revenue of the seller must be zero.
\begin{align*}
	\max_{G\in\Delta(I\times\Theta)} &\int_{I \times \Theta} U(p,\theta)\diff G(p,\theta)\tag{P}\label{eq:P} \\
	\text{ s.t. } & G(I\times \Theta')  =\mu(\Theta'),\forall \Theta' \subset\Theta\\
	& \int_{I'\times\Theta}R_p(p,\theta) \diff G(p,\theta) =0,\forall I'\subset I
\end{align*}

We use the following complementary slackness conditions that follow from the strong duality results of \citet{Kol18} and \citet{KCW24}, which relates problem~\eqref{eq:P} to the following dual problem.
\begin{align}
	\min_{{\lambda}, {\zeta}} &\int_{\Theta} {\lambda}(\theta) \diff \mu(\theta) \tag{D}\label{eq:D} \\
	\text{ s.t. } & 		{\lambda}(\theta)+{\zeta}(p)R_{p}(p,\theta) \geq U(p,\theta),\forall(p,\theta) \in I \times \Theta.\label{eq: strong duality condition 1}
\end{align}

\begin{lem}[\citealp{KCW24}]\label{lem: KCW24}
	Optimal solutions ${G}$ and ${\lambda},{\zeta}$ to the primal \eqref{eq:P} and the dual \eqref{eq:D} exist and their values are equal,
	\begin{align}
		\int_{I \times \Theta} U(p,\theta) \diff {G}(p,\theta) = \int_{\Theta} {\lambda}(\theta) \diff \mu(\theta).\label{eq: strong duality equivalence}
	\end{align}
\end{lem}

We now use this strong duality result to establish \autoref{prop: duality converted to our setting}.  By \autoref{lem: KCW24}, a feasible ${G}$ is an optimal solution to the primal problem \eqref{eq:P} if and only if there exists a feasible solution ${\lambda},{\zeta}$ to the dual problem \eqref{eq:D} such that \cref{eq: strong duality equivalence} holds.  For any feasible ${G}$, we have
\begin{align*}
	\int_{I \times \Theta} {\zeta}(p) R_p(p,\theta)\diff {G}(p,\theta) = 0,
\end{align*}
so
\begin{align*}
	\int_{I \times \Theta} U(p,\theta)\diff {G}(p,\theta) = 	\int_{I \times \Theta} [U(p,\theta) - {\zeta}(p) R_p(p,\theta)]\diff {G}(p,\theta),
\end{align*}
and
\begin{align*}
	 \int_{\Theta} {\lambda}(\theta) \diff \mu(\theta) = \int_{I \times \Theta} {\lambda}(\theta) \diff {G}(p,\theta).
\end{align*}
So \cref{eq: strong duality equivalence} is equivalent to
\begin{align*}
	\int_{I \times \Theta} [U(p,\theta) - {\zeta}(p) R_p(p,\theta)]\diff {G}(p,\theta) = \int_{I \times \Theta} {\lambda}(\theta) \diff {G}(p,\theta),
\end{align*}
which can be written as
	\begin{align}
{\lambda}(\theta)+{\zeta}(p)R_{p}(p,\theta) & =U(p,\theta),{G}\text{-almost surely} \label{eq: strong duality condition 2}	\end{align}

To summarize, we have shown that a probability measure ${G}$ is an optimal solution to the problem \eqref{eq:P}
	if and only if there exists functions ${\lambda}(\theta),{\zeta}(p)$
	such that \cref{eq: strong duality condition 1} and \cref{eq: strong duality condition 2} hold.  Recall that welfare-monotonicity is equivalent to the optimality of the no-information segmentation for \emph{all} prior distributions.  The no-information segmentation corresponds to a measure $G$ that assigns probability 1 to the prior market and its optimal price.  The no-information segmentation corresponds to a distribution $G$ that assigns probability 1 to price ${p}$ that is optimal for $\mu$ (and induces marginal $\mu$ over types).  So suppose that such a distribution $G$ is optimal for problem \eqref{eq:P}.  This observation pins down the function ${\lambda}$, because given \cref{eq: strong duality condition 2} for every $\theta$ we must have
\begin{align*}
	{\lambda}(\theta) =U({p},\theta) - {\zeta}({p})R_{p}({p},\theta).
\end{align*}
Substituting the definition of ${\lambda}$ into \cref{eq: strong duality condition 1}, we conclude that the no-information segmentation is optimal for a prior $\mu$ with optimal price ${p}$ if and only if there exists a function ${\zeta}$ such that for every type $\theta$,
	\begin{align*}
U({p},\theta) - {\zeta}({p})R_{p}({p},\theta) \geq U(p',\theta) - {\zeta}(p')R_{p}(p',\theta), \forall p',
\end{align*}
which we can write as
\begin{align}
	{p} \in \arg\max_{p'} U(p',\theta) - {\zeta}(p')R_{p}(p',\theta).\label{eq: rewriting IMB IMG using duality 1}
\end{align}

Recall that welfare-monotonicity holds if the no-information segmentation is optimal for \emph{every} prior $\mu$.  Notice that any price in $I$ is optimal for \emph{some} prior $\mu$.  Welfare-monotonicity therefore holds if and only if for every ${p}$ in $I$, there exists ${\zeta}$ such that \cref{eq: rewriting IMB IMG using duality 1} holds.  This function ${\zeta}$ may depend on ${p}$.  To make this dependence explicit, let us use $\zeta({p},\cdot)$ as a function that is indexed by ${p}$.   Using this notation, welfare-monotonicity holds if and only if there exists a single function $\zeta$ such that for every ${p} \in I$ and every type $\theta$,
\begin{align*}
	{p} \in \arg\max_{p'} U(p',\theta) - \zeta({p},p')R_{p}(p',\theta).
\end{align*}

\subsection{Proof of \autoref{thm:NSC}}\label{app: main result proof}
\subsubsection{Necessity of Partial Inclusion}\label{subsec:partial inclusion necessary}
	A simplifying step is that we can focus on two demands.  In particular, for a family of demands $\mathcal{D}$, and any two types $\theta_L,\theta_H$ in it, if $\alpha$-IMB ($\alpha$-IMG) does not hold for the binary family $\mathcal{D}' = \{D(p,\theta_L),D(p,\theta_H)\}$, then it does not hold for $\mathcal{D}$ either.
	
	Suppose first that there is full exclusion, that is, $\overline{p}(\theta_L) < p(\theta_H)$.  We want to show that both $\alpha$-IMB and $\alpha$-IMG must be violated.  It is sufficient to prove this claim for a certain (not every) prior market $\mu$ on $\theta_L,\theta_H$.  This is because for any prior market $\mu'$ we can first construct a segmentation that contains market $\mu$ in its support, and then show that segmenting $\mu'$ further does not increase or decrease weighted surplus (depending on which one of IMG or IMB we are proving).
	
	Consider the set $p(\mu)$ of optimal prices in a market in which $\theta$ has probability $1-\mu$ and $\theta_H$ has probability $\mu$,
	\begin{align*}
		p(\mu) = \arg \max_p (1-\mu)R(p,\theta_L) + \mu R(p,\theta_H).
	\end{align*}
	Notice that $p(\mu)$ never includes a price in $(\overline{p}(\theta_L),p(\theta_H))$ because any such price, which excludes $\theta_L$ entirely, leads to a lower revenue than $p(\theta_H)$.  Moreover, $p(\mu)$ never includes a price less than $p(\theta_L)$ because any such price is also lower than $p(\theta_H)$ and increasing it increases revenue in both types.  Therefore, $p(\mu)$ may only contain prices in $[p(\theta_L),\overline{p}(\theta_L)]\cup \{p(\theta_H)\}$.  At $\mu=0$, the only optimal price is $p(\theta_L)$, and at $\mu =1$, the only optimal price is $p(\theta_H)$.  Because for each $p$, revenue changes linearly in $\mu$, there exists a threshold $\hat{\mu} \in (0,1)$ such that $p(\theta_H) \in p(\mu)$ if and only if $\mu \geq \hat{\mu}$, and $p(\theta_H)$ is the unique optimal price for all $\mu > \hat{\mu}$.  
	
	To see that $\alpha$-IMB does not hold, consider any segmentation of the prior market that contains a segment in which the probability of $\theta_H$ in the prior market is $\mu' > \hat{\mu}$ (such a segmentation exists because the prior market has full support over $\theta_L,\theta_H$).  The unique optimal price for $\mu'$ is $p(\theta_H)$.  Consider segmenting $\mu'$ further into two segments $\mu_1 < \hat{\mu} < \mu'$ and $\mu_2 > \mu'$.  Because any optimal price for market $\mu_1$ is strictly less than $p(\theta_H)$ and the unique optimal price for market $\mu_2$ is $p(\theta_H)$, this segmentation increases consumer surplus.  Because any segmentation weakly increases producer surplus, it also strictly increases weighted surplus.
	
	To see that $\alpha$-IMG does not hold, consider any segmentation of the prior market that contains a segment in which the probability of $\theta_H$ in the prior market is $\mu' = \hat{\mu} - \epsilon$ for some small $\epsilon$.  Consider segmenting $\mu'$ into two segments $\mu_1 = \mu' - 2\epsilon$ and $\mu_2 = \mu' + 2\epsilon = \hat{\mu} + \epsilon$, with probability $\frac{1}{2}$ each.  As $\epsilon$ goes to zero, the seller's revenue approaches the optimal revenue in market $\mu'$ continuously by the Maximum theorem.  But we argue that the consumer surplus decreases discontinuously, that it, it decreases by at least some $\delta > 0$, and therefore weighted surplus decreases for small enough $\epsilon$.
	
	To see the discontinuity in consumer surplus, consider optimal prices for markets $\mu_1$ and $\mu_2$.  Because $\mu_1 < \hat{\mu}$, optimal prices for $\mu_1$ are those that maximize revenue over $[p(\theta_L),\overline{p}(\theta_L)]$.  Because the revenue curve is strictly concave over this range, the optimal price and therefore consumer surplus changes continuously.  But for market $\mu_2$ the only optimal price is $p(\theta_H)$ and the consumer surplus of this market is bounded away from that of $\mu'$, so the consumer surplus and total surplus decrease for small enough $\epsilon$.
	
	Now suppose there is full inclusion, $p(\theta_L) < \underline{p}(\theta_H)$.  We argue that there exist $\mu_1,\mu_2$ such that $\mu_1< \mu_2$ and price $\underline{p}(\theta_H)$ is the unique optimal price for any market in $(\mu_1,\mu_2)$.  For this, we first show that note that the revenue curve associated with $\theta_H$ has a kink at price $\underline{p}(\theta_H)$.  The revenue of type $\theta_H$ at any price $p < \underline{p}(\theta_H)$ is 
	\begin{align*}
		R(p,\theta_H) = p D(\underline{p}(\theta_H),\theta_H)
	\end{align*}
	so the left derivative of $R(p,\theta_H)$ at $p = \underline{p}(\theta)$ is
	\begin{align*}
		R_p(p,\theta_H) = D(\underline{p}(\theta),\theta_H).
	\end{align*}
	The revenue of type $\theta_H$ at any price $p > \underline{p}(\theta_H)$ is 
	\begin{align*}
			R(p,\theta_H) = p D(p,\theta_H)
	\end{align*}
	so marginal revenue is
	\begin{align*}
		R_p(p,\theta_H) = D(p,\theta_H) + pD_p(p,\theta_H).
	\end{align*}	
	As $p$ converges to $\underline{p}(\theta_H)$ from above, marginal revenue converges to
	\begin{align*}
		\lim_{p \rightarrow^+ \underline{p}(\theta_H)} R_p(p,\theta_H) = D(\underline{p}(\theta_H),\theta_H) + pD_p(\underline{p}(\theta_H),\theta_H) < D(\underline{p}(\theta_H),\theta_H).
	\end{align*}	
	So the right derivative of the revenue curve at $\underline{p}(\theta_H)$ is strictly less than its left derivative.  Let $\delta^- > \delta^+$ be the left and the right derivatives of the revenue curve of type $\theta_H$ at $\underline{p}(\theta_H)$.  Price $\underline{p}(\theta_H)$ is optimal in any market $\mu$ such that
	\begin{align*}
		(1-\mu)R_p(\underline{p}(\theta_H),\theta_L) + \mu \delta^- > 0, \text{ and } (1-\mu)R_p(\underline{p}(\theta_H),\theta_L) + \mu \delta^+ < 0,
	\end{align*}
	that is
	\begin{align*}
		\mu \in (\mu_1,\mu_2) := (\frac{-R_p(\underline{p}(\theta_H),\theta_L)}{\delta^+ - R_p(\underline{p}(\theta_H),\theta_L)},\frac{-R_p(\underline{p}(\theta_H),\theta_L)}{\delta^- -  R_p(\underline{p}(\theta_H),\theta_L)}).
	\end{align*}
	
	Now consider IMB.  consider any segmentation of the prior market that contains a segment in which the probability of $\theta_H$ in the prior market is $\mu' \in (\mu_1,\mu_2)$.  Consider segmenting $\mu'$ further into two segments $\mu'_1$ and $\mu'_2$ such that $\mu'_1 < \mu_1 < \mu'$ and $\mu'< \mu'_2 < \mu_2$.  Any optimal price for $\mu'_1$ is less than $\underline{p}(\theta_H)$ and the optimal price for $\mu'_2$ is $\underline{p}(\theta_H)$.  This segmentation therefore increases consumer surplus.  Because any segmentation weakly increases producer surplus, it also increases weighted surplus.
	
	Finally, consider IMG.  For this, we examine the value function $W^{\alpha}(\mu)$ around $\mu_2$ and show that the value function is locally concave, which means that providing a small amount of information reduces weighted surplus.  For this, consider the left derivative of $W^{\alpha}(\mu)$ at $\mu = \mu_2$.  Recall that price $\underline{p}(\theta_H)$ is optimal for all markets in $(\mu_1,\mu_2)$, which means that the left derivative of $W^{\alpha}(\mu)$ at $\mu < \mu_2$ is
	\begin{align*}
		W^{\alpha}_\mu(\mu) &=  V^\alpha(p(\mu),\theta_H) - V^\alpha(p(\mu),\theta_L) + \bigg((1-\mu) V^\alpha_p(p(\mu),\theta_L) + \mu V^\alpha_p(p(\mu),\theta_H)\bigg)p_{\mu}(\mu), \\
		&= V^\alpha(p(\mu),\theta_H) - V^\alpha(p(\mu),\theta_L),
	\end{align*}
	where the equality follows because $p_{\mu}(\mu) = 0$ for $\mu$ in $(\mu_1,\mu_2)$ as $p(\mu) = \underline{p}(\theta_H)$.  As $\mu$ converges to $\mu_2$ from below, this derivative converges to
	\begin{align*}
		V^\alpha(\underline{p}(\theta_H)),\theta_H) - V^\alpha(\underline{p}(\theta_H),\theta_L).
	\end{align*}
	Now consider what happens as $\mu$ converges to $\mu_2$ from above. The derivative of $W^{\alpha}$ is
	\begin{align*}
		V^\alpha(p(\mu),\theta_H) - V^\alpha(p(\mu),\theta_L)  + \bigg((1-\mu) V^\alpha_p(p(\mu),\theta_L) + \mu V^\alpha_p(p(\mu),\theta_H)\bigg)p_{\mu}(\mu).
	\end{align*}
	Given the first-order condition of the seller's problem, we can write this derivative as
	\begin{align*}
	V^\alpha(p(\mu),\theta_H) - V^\alpha(p(\mu),\theta_L)  + \alpha \bigg((1-\mu) CS_p(p(\mu),\theta_L) + \mu CS_p(p(\mu),\theta_H)\bigg)p_{\mu}(\mu).
\end{align*}	
	which, as $\mu$ goes to $\mu_2$ from above, converges to 
		\begin{align*}
		&V^\alpha(\underline{p}(\theta_H),\theta_H) - V^\alpha(\underline{p}(\theta_H),\theta_L)  + \alpha \bigg((1-\mu) CS_p(\underline{p}(\theta_H),\theta_L) + \mu CS_p(\underline{p}(\theta_H),\theta_H)\bigg)p_{\mu}(\mu_2) \\
		< &\ V^\alpha(\underline{p}(\theta_H),\theta_H) - V^\alpha(\underline{p}(\theta_H),\theta_L),
	\end{align*}
	where the inequality follows because $p(\mu)$ is strictly increasing at $\mu \geq \mu_2$, and $CS_p < 0$.   We conclude that $W^\alpha$ is not concave, which means that IMG does not hold.

\subsubsection{Statement (i)}
% Let us first give the formal connection between welfare-monotonicity, curvature of the value function, and optimality of the no-information segmentation for all priors.

% \begin{lem}\label{lem: equivalence between IMB and NIS and convexity}
% The following three statements are equivalent for any $(\mathcal{D},\mu)$.
% \begin{enumerate}
% 	\item $\alpha$-IMB holds for $(\mathcal{D},\mu)$.
% 	\item The value function $W^{\alpha}$ is concave.
% 	\item No-information is ``$V^{\alpha}$ universally optimal'', that is, for every prior $\mu'$ on $\mathcal{D}$, the no-information segmentation that assigns probability 1 to $\mu'$ solves 
% \begin{align*}
% 	&\max_{\sigma \in S(\mu')} V^{\alpha}(\sigma).
% \end{align*}
% \end{enumerate}
% Furthermore, the following three statements are equivalent for any $(\mathcal{D},\mu)$.
% \begin{enumerate}[label={\arabic*\hspace{1pt}$'$.}]
% 	\item $\alpha$-IMG holds for $(\mathcal{D},\mu)$.
% 	\item The value function $W^{\alpha}$ is convex.
% 	\item No-information is ``$-V^{\alpha}$ universally optimal'', that is, for every prior $\mu'$ on $\mathcal{D}$, the no-information segmentation that assigns probability 1 to $\mu'$ solves 
% 	\begin{align*}
% 		&\max_{\sigma \in S(\mu')} -V^{\alpha}(\sigma).
% 	\end{align*}
% \end{enumerate}
% \end{lem}

Let us start by an implication of the spanning property that will be later used in the proof.

\begin{lem}\label{lem: spanning consequence}
    If the spanning condition of \autoref{thm:NSC} is violated, then there exists some $\theta$, $p\in I$, and $a,b$ such that $R_p(p,\theta) \neq 0$ and
\begin{align*}
R_p(p,\theta) & =a R_p(p,\theta_L)+b R_p(p,\theta_H)\\
R_{pp}(p,\theta) & =a R_{pp}(p,\theta_L)+b R_{pp}(p,\theta_H)\\
U_p(p,\theta) & \neq a U_p(p,\theta_L)+b U_p(p,\theta_H).
\end{align*}
\end{lem}
\begin{proof}
    We will show that there is an open interval of prices $I' \subset (p(\theta_L),p(\theta_H))$ such that the above three equations hold for all $p \in I'$.  Because marginal revenue is decreasing in $p$, there exists such a price that furthermore satisfies $R_p(p,\theta) \neq 0$.

    For each $p$ there exists a pair $a,b$ of constants for which the first two equalities hold, that is, writing $a,b$ as functions of $p$ to make this dependence explicit,
 \begin{align}
R_p(p,\theta) & =a(p) R_p(p,\theta_L)+b(p) R_p(p,\theta_H)\label{eq: spanning lemma WTS neq 0 1}\\
R_{pp}(p,\theta) & =a(p) R_{pp}(p,\theta_L)+b(p) R_{pp}(p,\theta_H),\label{eq: spanning lemma WTS neq 0 2}
\end{align}   
which is obtained by solving the system of two equations and two variables,
\begin{align*}
a(p) & =\frac{R_{pp}(p,\theta)R_{p}(p,\theta_H)-R_{p}(p,\theta)R_{pp}(p,\theta_H)}{R_{p}(p,\theta_H)R_{pp}(p,\theta_H)-R_{p}(p,\theta_H)R_{pp}(p,\theta_H)}\\
b(p) & =\frac{-R_{pp}(p,\theta)R_{p}(p,\theta_H)+R_{p}(p,\theta)R_{pp}(p,\theta_H)}{R_{p}(p,\theta_H)R_{pp}(p,\theta_H)-R_{p}(p,\theta_H)R_{pp}(p,\theta_H)}.
\end{align*}
This solution is well-defined because $R_{pp}<0$ and $R_{p}(p,\theta_H)<0<R_{p}(p,\theta_H)$.
    
    Similarly, let us define functions $c,d$ as follows
\begin{align*}
R(p,\theta) & =c(p) R(p,\theta_H)+d(p) R(p,\theta_H)\\
R_{p}(p,\theta) & =c(p) R_{p}(p,\theta_H)+d(p) R_{p}(p,\theta_H)
\end{align*}
Using a similar argument to above, the pair $c(p),d(p)$ is also unique for each $p$.
Violation of the spanning condition implies that $c,d$
are not constant functions.

Taking a derivative of the above, we have
\begin{align}
R_{p}(p,\theta)= & c'(p)R(p,\theta_H)+d'(p)R_{}(p,\theta_H)+c(p) R_{p}(p,\theta_H)+d(p) R_{p}(p,\theta_H) \nonumber \\
= & c(p) R_{p}(p,\theta_H)+d(p) R_{p}(p,\theta_H),\label{eq: spanning lemma WTS neq 0 3}
\end{align}
which means
\begin{align*}
c'(p)R(p,\theta_H)=-d'(p)R(p,\theta_H),
\end{align*}
and
\begin{align}
R_{pp}(p,\theta)= & c'(p)R_{p}(p,\theta_H)+d'(p)R_{p}(p,\theta_H)+c(p) R_{pp}(p,\theta_H)+d(p) R_{pp}(p,\theta_H)\nonumber \\
= & (\frac{R_{p}(p,\theta_H)}{R(p,\theta_H)}-\frac{R_{p}(p,\theta_H)}{R(p,\theta_H)})c'(p)R(p,\theta_H)+c(p) R_{pp}(p,\theta_H)+d(p) R_{pp}(p,\theta_H) \nonumber \\
= & \Delta(p)+c(p) R_{pp}(p,\theta_H)+d(p) R_{pp}(p,\theta_H),\label{eq: spanning lemma WTS neq 0 4}
\end{align}
where $\Delta(p)$ is defined as
\begin{align*}
\Delta(p) := (\frac{R_{p}(p,\theta_H)}{R(p,\theta_H)}-\frac{R_{p}(p,\theta_H)}{R(p,\theta_H)})c'(p)R(p,\theta_H).
\end{align*}
Let $I'$ be an interval for which $c'(p)\neq0$,
which means that $\Delta(p)\neq0$ for the same interval.

Combining \cref{eq: spanning lemma WTS neq 0 1} and \cref{eq: spanning lemma WTS neq 0 2} with \cref{eq: spanning lemma WTS neq 0 3} and \cref{eq: spanning lemma WTS neq 0 4}, we have
\begin{align}
a(p) & =\frac{R_{pp}(p,\theta)R_{p}(p,\theta_H)-R_{p}(p,\theta)R_{pp}(p,\theta_H)}{R_{p}(p,\theta_H)R_{pp}(p,\theta_H)-R_{p}(p,\theta_H)R_{pp}(p,\theta_H)}\nonumber \\
 & =c(p)+\frac{\Delta(p) R_{p}(p,\theta_H)}{R_{p}(p,\theta_H)R_{pp}(p,\theta_L)-R_{p}(p,\theta_L)R_{pp}(p,\theta_H)} \label{eq: spanning lemma WTS neq 0 5}\\
b(p) & =\frac{-R_{pp}(p,\theta)R_{p}(p,\theta_H)+R_{p}(p,\theta)R_{pp}(p,\theta_H)}{R_{p}(p,\theta_H)R_{pp}(p,\theta_H)-R_{p}(p,\theta_H)R_{pp}(p,\theta_H)} \nonumber \\
 & =d(p)-\frac{\Delta(p) R_{p}(p,\theta_L)}{R_{p}(p,\theta_H)R_{pp}(p,\theta_L)-R_{p}(p,\theta_L)R_{pp}(p,\theta_H)}.\label{eq: spanning lemma WTS neq 0 6}
\end{align}

Now let us evaluate 
\begin{align}
U_p(p,\theta) - a(p) U_p(p,\theta_L) - b(p) U_p(p,\theta_H)\label{eq: spanning lemma WTS neq 0}
\end{align}
using the above two equations and show that it is not equal to zero.  From the definition of $U$, when characterizing $\alpha$-IMG the above expression is
\begin{align*}
U_p(p,\theta) - a U_p(p,\theta_L) - b U_p(p,\theta_H) &= \alpha \Big(CS_p(p,\theta) - a(p) CS_p(p,\theta_L) - b(p) CS_p(p,\theta_H)\Big) \\
&+ (1-\alpha) \Big(R_p(p,\theta) - a(p) R_p(p,\theta_L) - b(p) R_p(p,\theta_H)\Big) \\
&= -\alpha \Big(D(p,\theta) - a(p) D(p,\theta_L) - b(p) D(p,\theta_H)\Big)\\
&= \frac{-\alpha}{p} \Big(R(p,\theta) - a(p) R(p,\theta_L) - b(p) R(p,\theta_H) \Big).
\end{align*}
When characterizing $\alpha$-IMB, we want to show that the negative of the above expression is non-zero.  In either case, because $\alpha > 0$, to show that \cref{eq: spanning lemma WTS neq 0} is not zero, it is sufficient to show that
\begin{align*}
    R(p,\theta) - a(p) R(p,\theta_L) - b(p) R(p,\theta_H) \neq 0.
\end{align*}

For this, use \cref{eq: spanning lemma WTS neq 0 5} and \cref{eq: spanning lemma WTS neq 0 6} to write
\begin{align*}
& R(p,\theta)-a(p)R(p,\theta_L)-b(p)R(p,\theta_H) \\
 = & R(p,\theta)-c(p) R(p,\theta_L)-d(p) R(p,\theta_H)-\Delta(p) \frac{R_{p}(p,\theta_H)R(p,\theta_L)-R_{p}(p,\theta_L)R(p,\theta_H)}{R_{p}(p,\theta_H)R_{pp}(p,\theta_L)-R_{p}(p,\theta_L)R_{pp}(p,\theta_H)}\\
 = & -\Delta(p) \frac{R_{p}(p,\theta_H)R(p,\theta_L)-R_{p}(p,\theta_L)R(p,\theta_H)}{R_{p}(p,\theta_H)R_{pp}(p,\theta_L)-R_{p}(p,\theta_L)R_{pp}(p,\theta_H)}.
\end{align*}
For any $p \in I'$, this expression is non-zero because $\Delta(p) \neq 0$, which proves the claim.
\end{proof}

Given \autoref{lem: spanning consequence}, we now complete the proof of statement (i) of \autoref{thm:NSC}.

\paragraph{Proof of statement (i).} We have already shown in \cref{subsec:partial inclusion necessary} that partial inclusion is necessary for the welfare-monotonicity properties.  We therefore assume partial inclusion here and establish the necessity and sufficiency of the reduction.

The case where $\min_\theta p(\theta) = \max_\theta p(\theta)$ is straightforward.  In this case, both welfare-monotonicity properties as well as the conditions of the statement hold trivially.  So suppose for the rest of the proof that $\min_\theta p(\theta) < \max_\theta p(\theta)$.

\paragraph{Necessity.} Suppose there is partial inclusion.  Recall from \autoref{prop: duality converted to our setting} that welfare-monotonicity holds if and only if there exists a function $\zeta$ such that for every ${p} \in I$ and every type $\theta$,
\begin{align}
	{p} \in \arg\max_{p'\in I} U(p',\theta) - \zeta({p},p')R_{p}(p',\theta).\label{eq: strong duality repeated}
\end{align}

First, notice that if \cref{eq: strong duality repeated} is satisfied for all $\theta \in \Theta$, then it must be satisfied for $\theta_L,\theta_H$ that have the lowest and the highest monopoly price in $\Theta$, and therefore welfare-monotonicity holds for the binary family that consists only of $\theta_L,\theta_H$.

Now consider any pair $\theta,\theta'$ of types for which $p(\theta) < p(\theta')$ and any ${p} \in (p(\theta), p(\theta'))$.  \cref{eq: strong duality repeated} implies
\[
U(p',\theta)-\zeta({p},p')R_p(p',\theta)\leq U({p},\theta)-\zeta({p},{p})R_p({p},\theta),\forall p'\in (p(\theta), p(\theta')).
\]
Because $R_p(p',\theta)<0$, we can divide the above inequality by $R_p(p',\theta)$
and write it as
\[
\zeta({p},p')\leq\frac{U(p',\theta)-U({p},\theta)+\zeta({p},{p})R_p({p},\theta)}{R_p(p',\theta)}.
\]
A similar argument repeated for $\theta'$, but with the difference that $R_p(p',\theta')>0$, implies
\[
\zeta({p},p')\geq\frac{U(p',\theta')-U({p},\theta')+\zeta({p},{p})R_p({p},\theta')}{R_p(p',\theta')}.
\]
Therefore, we must have that for all $p,p'\in (p(\theta), p(\theta'))$:
\[
\frac{U(p',\theta)-U({p},\theta)+\zeta({p},{p})R_p({p},\theta)}{R_p(p',\theta)} \geq \frac{U(p',\theta')-U({p},\theta')+\zeta({p},{p})R_p({p},\theta')}{R_p(p',\theta')}.
\]
Note that evaluated at $p'={p}$, both sides of the above
are equal to $\zeta({p},{p})$. Because both sides are
continuously differentiable, it has to be that they are tangent at
$p'={p}$. Therefore, 
\[
\frac{U_p({p},\theta)}{R_p({p},\theta)}-\zeta({p},{p}) \frac{R_{pp}({p},\theta)}{R_p({p},\theta)} = \frac{U_p({p},\theta')}{R_p({p},\theta')}-\zeta({p},{p}) \frac{R_{pp}({p},\theta')}{R_p({p},\theta')}
\]
which pins down $\zeta({p},{p})$,
\begin{align}
    \zeta({p},{p}) = \frac{\frac{U_p({p},\theta)}{R_p({p},\theta)}-\frac{U_p({p},\theta')}{R_p({p},\theta')}}{\frac{R_{pp}({p},\theta)}{R_p({p},\theta)}-\frac{R_{pp}({p},\theta')}{R_p({p},\theta')}} = \frac{U_p({p},\theta)R_p({p},\theta')-U_p({p},\theta')R_p({p},\theta)}{R_{pp}({p},\theta)R_p({p},\theta')-R_{pp}({p},\theta')R_p({p},\theta)}.
\end{align}

Now suppose the spanning condition of the theorem is violated, so there exists some $\theta$ whose demand curve cannot be written as a linear combination of $D(\cdot,\theta_L)$ and $D(\cdot,\theta_H)$ over the interval $I$.  \autoref{lem: spanning consequence} implies that that there is price ${p}\in I$ and $a,b$ such that $R_p({p},\theta) \neq 0$ and
\begin{align*}
R_p({p},\theta) & =a R_p({p},\theta_L)+b R_p({p},\theta_H)\\
R_{pp}({p},\theta) & =a R_{pp}({p},\theta_L)+b R_{pp}({p},\theta_H)\\
U_p({p},\theta) & \neq a U_p({p},\theta_L)+b U_p({p},\theta_H).
\end{align*}

Because $R_p({p},\theta_L) < 0 < R_p({p},\theta_H)$ and $R_p({p},\theta) \neq 0$, the marginal revenue of $\theta$ has a different sign with either $\theta_L,\theta_H$.  Suppose $0 < R_p({p},\theta)$ (the other case is similar).  Then our discussion above pins down $\zeta({p},{p})$ in two different ways
\begin{align*}
    \zeta({p},{p}) = \frac{U_p({p},\theta_L)R_p({p},\theta_H)-U_p({p},\theta_H)R_p({p},\theta_L)}{R_{pp}({p},\theta_L)R_p({p},\theta_H)-R_{pp}({p},\theta_H)R_p({p},\theta_L)}
\end{align*}
and
\begin{align*}
    \zeta({p},{p}) =\frac{U_p({p},\theta_L)R_p({p},\theta)-U_p({p},\theta)R_p({p},\theta_L)}{R_{pp}({p},\theta_L)R_p({p},\theta)-R_{pp}({p},\theta)R_p({p},\theta_L)}.
\end{align*}

But these two expressions cannot be equal by \autoref{lem: spanning consequence}, implying that $\zeta$ satisfying \cref{eq: strong duality repeated} does not exist.  To see this, let us use \autoref{lem: spanning consequence} and the above two equalities to write
\begin{align*}
    \zeta({p},{p}) &= \frac{U_p({p},\theta_L)R_p({p},\theta)-U_p({p},\theta)R_p({p},\theta_L)}{R_{pp}({p},\theta_L)R_p({p},\theta)-R_{pp}({p},\theta)R_p({p},\theta_L)} \\
    &\neq \frac{U_p({p},\theta_L)\Big(a R_p({p},\theta_L)+b R_p({p},\theta_H)\Big)-\Big(a U_p({p},\theta_L)+b U_p({p},\theta_H)\Big) R_p({p},\theta_L)}{R_{pp}({p},\theta_L)\Big(a R_p({p},\theta_L)+b R_p({p},\theta_H)\Big)-\Big(a R_{pp}({p},\theta_L)+b R_{pp}({p},\theta_H)\Big) R_p({p},\theta_L)} \\
    &= \frac{b\Big(U_p({p},\theta_L)R_p({p},\theta_H)-U_p({p},\theta_H)R_p({p},\theta_L)\Big)}{b\Big(R_{pp}({p},\theta_L)R_p({p},\theta_H)-R_{pp}({p},\theta_H)R_p({p},\theta_L)\Big)} \\
    &= \zeta({p},{p}),
\end{align*}
which is a contradiction.

\paragraph{Sufficiency.}  We now prove the sufficiency part of the reduction.  Suppose a family of demand curves $\mathcal{D}$ can be decomposed into two demands that satisfy $\alpha$-IMG ($\alpha$-IMB),
\begin{align*}
	D(p,\theta) & =f_{1}(\theta)D(p,\theta_L)+f_{2}(\theta)D(p,\theta_H), \forall p,\theta.
\end{align*}
We want to show that $\mathcal{D}$ satisfies $\alpha$-IMG ($\alpha$-IMB).  That is, there exists a function $\zeta$ such that for every price ${p} \in I$, we have
\begin{align*}
	{p} \in \arg\max_{p'\in I} U(p',\theta) - \zeta({p},p')R_{p}(p',\theta),
\end{align*}
for $W = V^{\alpha}$ ($W = -V^{\alpha}$).

Notice that because each demand is a linear combination of the two basis demands, each value, revenue, and marginal revenue function can also be written using a linear combination of the corresponding objects for the basis demand curves. Formally,
\begin{align*}
	U(p,\theta) & =f_{1}(\theta)U(p,\theta_L)+f_{2}(\theta)U(p,\theta_H), \\
	R(p,\theta) & =f_{1}(\theta)R(p,\theta_L)+f_{2}(\theta)R(p,\theta_H),\\
	R_p(p,\theta) & =f_{1}(\theta)R_p(p,\theta_L)+f_{2}(\theta)R_p(p,\theta_H).
\end{align*}

Because the family $\{D(p,\theta_L),D(p,\theta_H)\}$ satisfies $\alpha$-IMG ($\alpha$-IMB), for each ${p}\in I$ we have
\begin{align*}
	{p} \in \arg\max_{p' \in I} U(p',\theta_L) - \zeta({p},p')R_{p}(p',\theta_L),\\
	{p} \in \arg\max_{p' \in I} U(p',\theta_H) - \zeta({p},p')R_{p}(p',\theta_H).
\end{align*}
Because ${p}$ maximizes each of the above two expressions, it also maximizes their linear combination,
\begin{align*}
	{p} &\in \arg\max_{p' \in I} f_1(\theta)\bigg(U(p',\theta_L) - \zeta({p},p')R_{p}(p',\theta_L)\bigg)+f_2(\theta) \bigg(U(p',\theta_H) - \zeta({p},p')R_{p}(p',\theta_H)\bigg) \\
	&= \arg\max_{p' \in I} U(p',\theta) - \zeta({p},p')R_{p}(p',\theta)
\end{align*}
as claimed, completing the proof.

\subsubsection{Statement (ii)}
We have already shown in \cref{subsec:partial inclusion necessary} that partial inclusion is necessary for the welfare-monotonicity properties.  
	So suppose there is partial inclusion.   By the concavification result of \citet{KaG11}, $\alpha$-IMB ($\alpha$-IMG) holds if and only if value as a function of $\mu$,
	\begin{align*}
		W^{\alpha}(\mu) = (1-\mu) V^\alpha(p(\mu),\theta_L) + \mu V^\alpha(p(\mu),\theta_H),
	\end{align*}
	is concave where $p(\mu)$ is the profit-maximizing price for the seller.  The partial inclusion condition implies that this price is uniquely defined and $p(\mu) \in [p(\theta_L),p(\theta_H)]$ is the profit-maximizing price for the seller that solves the seller's first-order condition
	\begin{align*}
		(1-\mu) R_p(p(\mu),\theta_L) + \mu R_p(p(\mu),\theta_H) = 0.
	\end{align*}
	
	We identify three consequences of the first-order condition for future use.  First, re-arranging the first-order condition, we have
	\begin{align}
		\frac{\mu}{1-\mu} = \frac{- R_p(p(\mu),\theta_L)}{R_p(p(\mu),\theta_H)}.\label{eq: two types concavity eq1}
	\end{align}
	Second, taking the derivative of the first-order condition with respect to $\mu$, we have
	\begin{align*}
		- R_p(p(\mu),\theta_L) + R_p(p(\mu),\theta_H) + \bigg((1-\mu) R_{pp}(p(\mu),\theta_L) + \mu R_{pp}(p(\mu),\theta_H)\bigg) p_{\mu}(\mu) = 0,
	\end{align*}
	which, after re-arranging, means
	\begin{align}
		p_{\mu}(\mu) = \frac{R_p(p(\mu),\theta_L) - R_p(p(\mu),\theta_H)}{(1-\mu) R_{pp}(p(\mu),\theta_L) + \mu R_{pp}(p(\mu),\theta_H)}.\label{eq: two types concavity eq2}
	\end{align}
	Third, $p(\mu)$ is increasing in $\mu$.  For this, notice that because $p(\theta_L) \leq p(\mu) \leq p(\theta_H)$ and each revenue curve is concave, the first term in the numerator, $R_p(p(\mu),\theta_L)$, is negative and the second term, $R_p(p(\mu),\theta_H)$, is positive, so the numerator is negative. 
    %{To be precise, at the two extremes one term might be zero but then the other term is non-zero.  At $p(\theta_L)$ the first term is zero but the second term is positive, so the expression is negative.  Similarly at $p(\theta_H)$ the second term is zero but the first term is negative, so the expression is negative.}  
    Further, because both revenue curves are strictly concave, the denominator is negative.
	
	Now $W^{\alpha}(\mu)$ is concave if and only if its derivative $W^{\alpha}_\mu(\mu)$ is decreasing.  From the definition of  $W^{\alpha}(\mu)$, and using \cref{eq: two types concavity eq1} and \cref{eq: two types concavity eq2}, we can write its derivative $W^{\alpha}_\mu(\mu)$ as
	\begin{align*}
		W^{\alpha}_\mu(\mu) = & V^\alpha(p(\mu),\theta_H) - V^\alpha(p(\mu),\theta_L)   \\ &+ \frac{V^\alpha_p(p(\mu),\theta_L) - \frac{ R_p(p(\mu),\theta_L)}{R_{p}(p(\mu),\theta_H)} V^\alpha_p(p(\mu),\theta_H)}{R_{pp}(p(\mu),\theta_L) - \frac{R_p(p(\mu),\theta_L)}{R_p(p(\mu),\theta_H)} R_{pp}(p(\mu),\theta_H)} \bigg(R_p(p(\mu),\theta_L) - R_p(p(\mu),\theta_H) \bigg).
	\end{align*}
	
	Because $p(\mu)$ is increasing in $\mu$, $W^{\alpha}_\mu(\mu)$ is decreasing if and only if $W^{\alpha}_\mu(\mu(p))$ is decreasing, where $\mu$ is the inverse of $p$, that is, $p(\mu(p)) = p$.  So we want 
	\begin{align*}
		& V^\alpha(p,\theta_H) - V^\alpha(p,\theta_L)   + \frac{V^\alpha_p(p,\theta_L) - \frac{ R_p(p,\theta_L)}{R_p(p,\theta_H)} V^\alpha_p(p,\theta_H)}{R_{pp}(p,\theta_L) - \frac{R_p(p,\theta_L)}{R_p(p,\theta_H)} R_{pp}(p,\theta_H)} \Big(R_p(p,\theta_L) - R_p(p(\mu),\theta_H) \Big)
	\end{align*}
	to be decreasing in $p$, as claimed.

%%%%%%%%%%%%%%%%%%%%%%%%%%%%%%%%%%%%%%%%%%%

\subsection{Proof of \autoref{thm:bounds}}\label{app: main result 2 proof}
\subsubsection{Statement (i)}
In addition to proving statement (i) of the theorem, stating that the eigenvalues and eigenvectors of the Hessian of the value function exist and can be calculated in closed form, we explicitly derive the closed forms as discussed in \autoref{sec: general closed form}.

The value function is
\begin{align*}
    W^{\alpha}(\mu) &= \sum_{k=2}^n V^{\alpha}(\theta_k,p(\mu))\mu_k + (1-\sum_{k=2}^n\mu_k) V^{\alpha}(\theta_L,p(\mu)) \\
    &= \sum_{k=2}^n \Delta V^{\alpha}(\theta_k,p(\mu))\mu_k +  V^{\alpha}(\theta_1,p(\mu)),
\end{align*}
The derivative with respect to $\mu_i$ and $\mu_j$ is
\begin{align*}
    W^{\alpha}_{i,j}(\mu)
    &= \sum_{k=1}^n V^{\alpha}_{pp}(\theta_k,p(\mu))p_j(\mu)p_i(\mu)\mu_k \\
    &+ \sum_{k=1}^n V^{\alpha}_p(\theta_k,p(\mu))p_{i,j}(\mu)\mu_k \\ & + \Delta V^{\alpha}_p(\theta_j,p(\mu))p_i(\mu) +  \Delta V^{\alpha}_p(\theta_i,p(\mu))p_j(\mu).
\end{align*}
In matrix form, the Hessian is
\begin{align}
    \nabla^2 W^{\alpha}(\mu) &= \mathbb{E}[V^{\alpha}_{pp}(\theta,p(\mu))] \nabla p(\mu)^T \nabla p(\mu) \nonumber \\
    &+ \mathbb{E}[V^{\alpha}_{p}(\theta,p(\mu))] \nabla^2 p(\mu) \nonumber \\
    &= (\Delta V^{\alpha}_p(\mu))^T \nabla p(\mu) + \nabla p(\mu)^T \Delta V^{\alpha}_p(\mu).\label{eq: Hessian calculation appendix}
\end{align}

To see that $p_{i,j}(\mu)$ can be written in a separable form, recall the first order condition for the optimal price,
\begin{align*}
    0 = \sum_{k=2}^n R_p(\theta_k,p(\mu))\mu_k + (1-\sum_{k=2}^n\mu_k) R_p(\theta_1,p(\mu)).
\end{align*}
Taking the derivative with respect to $\mu_i$ and then $\mu_j$ gives,
\begin{align*}
    0 &= \sum_{k=1}^n R_{ppp}(\theta_k,p(\mu))p_j(\mu)p_i(\mu)\mu_k \\
    &+ \sum_{k=1}^n R_{pp}(\theta_k,p(\mu))p_{i,j}(\mu)\mu_k \\ & + \Delta R_{pp}(\theta_j,p(\mu))p_i(\mu) +  \Delta R_{pp}(\theta_i,p(\mu))p_j(\mu),
\end{align*}
and therefore, in matrix form,
\begin{align}
    \nabla^2 p(\mu) = -\frac{\mathbb{E}[R_{ppp}(\theta,p(\mu))] \nabla p(\mu)^T \nabla p(\mu) + (\Delta R_{pp}(\mu))^T \nabla p(\mu) + \nabla p(\mu)^T \Delta R_{pp}(\mu)}{\mathbb{E}[R_{pp}(\theta,p(\mu))]}\label{eq: Hessian p is separable}
\end{align}
Combining \cref{eq: Hessian calculation appendix} and \cref{eq: Hessian p is separable} gives the separable form of the Hessian in \cref{eq: hessian separated}.  The Hessian matrix has rank at most 2, which means it has at most two non-zero eigenvalues.  We have given the explicit expression for two eigenvalues, which proves that they are the lowest and the highest eigenvalues.

\subsubsection{Statement  (ii) and (iii)}
Consider segmentations $\sigma,\sigma'$ with random markets $\mu,\mu'$ where  $\sigma$ is a mean-preserving spread of $\sigma'$. By Blackwell's theorem, there
exists a probability distribution function $Q\left(\cdot|\mu'\right)\in\Delta\Delta\Theta$ such that 
\begin{align*}
\sigma\left(A\right) & =\int Q\left(A|\mu'\right)\diff\sigma', \forall \text{ Borel } A\subset\Delta\Theta \\
\mu' & =\int\mu \diff Q\left(\mu|\mu'\right),\forall\mu'\in\Delta\Theta.
\end{align*}
This implies that 
\begin{align*}
V^{\alpha}\left(\sigma\right)-V^{\alpha}\left(\sigma'\right) & =\int W^{\alpha}\left(\mu\right)\diff\sigma-\int W^{\alpha}\left(\mu'\right)\diff\sigma'\\
 & =\int\int W^{\alpha}\left(\mu\right)\diff Q\left(\mu|\mu'\right)\diff\sigma'-\int W^{\alpha}\left(\mu'\right)\diff\sigma'\\
 & =\int\int\left[W^{\alpha}\left(\mu\right)-W^{\alpha}\left(\mu'\right)\right]\diff Q\left(\mu|\mu'\right)\diff \sigma'.
\end{align*}

Since $W^{\alpha}$ is a twice differentiable function, we can use
the Cauchy form of the Taylor expansion of $W$ and write
\[
W^{\alpha}\left(\mu\right)-W^{\alpha}\left(\mu'\right)=\left(\mu-\mu'\right)^{T}\nabla W^{\alpha}\left(\mu\right)+\frac{1}{2}\left(\mu-\mu'\right)^{T}\nabla^{2}W^{\alpha}\left(\tilde{\mu}\right)\left(\mu-\mu\right)
\]
where $\tilde{\mu}$ is some belief on the line that connects $\mu$
to $\mu'$. Integrating over the above for a $\mu'$, we must have that 
\begin{align*}
\int\left[W^{\alpha}\left(\mu\right)-W^{\alpha}\left(\mu'\right)\right]\diff Q\left(\mu|\mu'\right) & =\int\left(\mu-\mu'\right)^{T}\nabla W^{\alpha}\left(\mu\right)\diff Q\left(\mu|\mu'\right)\\
 & +\int\frac{1}{2}\left(\mu-\mu'\right)^{T}\nabla^{2}W^{\alpha}\left(\tilde{\mu}\right)\left(\mu-\mu'\right)\diff Q\left(\mu|\mu'\right)\\
 & =\int\frac{1}{2}\left(\mu-\mu'\right)^{T}\nabla^{2}W^{\alpha}\left(\tilde{\mu}\right)\left(\mu-\mu'\right)\diff Q\left(\mu|\mu'\right)
\end{align*}
where in the above we have used the fact that the average value of
$\mu$ under $Q\left(\cdot|\mu'\right)$ is $\mu'$. Because the Hessian matrix is symmetric, we can use the Courant–Fischer–Weyl min-max principle to bound the integrand using the highest eigenvalue $\overline{\lambda}\left(\tilde{\mu}\right)$ of $\nabla^{2}W\left(\tilde{\mu}\right)$,
\begin{align}
\left(\mu-\mu'\right)^{T}\nabla^{2}W^{\alpha}\left(\tilde{\mu}\right)\left(\mu-\mu'\right) & \leq\overline{\lambda}\left(\tilde{\mu}\right)\left(\mu-\mu'\right)^{T}\left(\mu-\mu'\right) \label{eq: bounds eq1} \\
 & \leq\max_{\mu''\in\Delta\Theta}\overline{\lambda}\left(\mu''\right)\left(\mu-\mu'\right)^{T}\left(\mu-\mu'\right), \label{eq: bounds eq2} 
\end{align}
Using the above inequality, we must have that 
\begin{align*}
\int\left[W^{\alpha}\left(\mu\right)-W^{\alpha}\left(\mu'\right)\right]\diff Q\left(\mu|\mu'\right) & \leq\frac{1}{2}\max_{\mu''\in\Delta\Theta}\overline{\lambda}\left(\mu''\right)\int\left(\mu-\mu'\right)^{T}\left(\mu-\mu'\right)\diff Q\left(\mu|\mu'\right)\\
 & =\frac{1}{2}\max_{\mu''\in\Delta\Theta}\overline{\lambda}\left(\mu''\right)\int\sum_{i\geq2}\left[\left(\mu_{i}\right)^{2}-\left(\mu'_{i}\right)^{2}\right]\diff Q\left(\mu|\mu'\right)
\end{align*}
Integrating over the above establishes the upper bound. The derivation
of the lower bound follows the same logic, inverting the two inequalities \cref{eq: bounds eq1} and \cref{eq: bounds eq2} and replacing max with min.

To see that the bounds are tight, note that we can always
get arbitrarily close to $\max_{\mu''\in\Delta\Theta}\overline{\lambda}\left(\mu''\right)$
with full-support markets. Let $\hat{\mu}$ be a full-support market that achieves that maximum minus $\epsilon$. If we let $\sigma'$
be the no information segmentation associated with $\hat{\mu}$, as
$\sigma$ converges to $\sigma'$, $\tilde{\mu}$ in the Taylor expansion of $W^{\alpha}$
converges to $\hat{\mu}$. Since $\hat{\mu}$ is interior, it can
be perturbed in all directions including the direction that achieves
the maximum value.  This proves that the bound is tight with a gap of at most $\epsilon$ for any arbitrary $\epsilon$.

Finally, we prove statement (iii) regarding the optimal direction. For any segmentation $\sigma_{\epsilon}$, the Cauchy form of the Taylor expansion of $W$ implies that
\begin{align}
    V^{\alpha}(\sigma_{\epsilon}) - W^{\alpha}(\mu_0)  & =\int\left[W^{\alpha}\left(\mu\right)-W^{\alpha}\left(\mu_0\right)\right]\diff \sigma_{\epsilon} \nonumber \\
    &= \int\left[W^{\alpha}\left((1-\epsilon)\mu_0+\epsilon \mu\right)-W^{\alpha}\left(\mu_0\right)\right]\diff \sigma \nonumber \\
    &=\int\frac{\epsilon^2}{2}\left(\mu-\mu_0\right)^{T}\nabla^{2}W^{\alpha}\left(\tilde{\mu}_{\epsilon}\right)\left(\mu-\mu_0\right)\diff \sigma \label{eq: best direction eq1}
\end{align}
where $\tilde{\mu}_{\varepsilon}$ is a market on the line that connects
$(1-\epsilon)\mu_0+\epsilon \mu$ to $\mu_0$. Note that
in the above as $\varepsilon\rightarrow0$, $\tilde{\mu}_{\varepsilon}\rightarrow \mu_0$
and hence, we can write
\begin{align}
\lim_{\varepsilon\rightarrow0}\frac{V^{\alpha}(\sigma_{\epsilon}) - W^{\alpha}(\mu_0)}{\varepsilon^{2}}&=\frac{1}{2}\int\left(\mu-\mu_0\right)^{T}\nabla^{2}W\left(\mu_0\right)\left(\mu-\mu_0\right)\diff \sigma \nonumber \\
&\leq \frac{1}{2}\int \overline{\lambda}\left({\mu_0}\right)\left(\mu-\mu_0\right)^{T}\left(\mu-\mu_0\right) \diff \sigma \nonumber  \\
&= \frac{1}{2}\overline{\lambda}(\mu_0) \left[\mathbb{E}_{\sigma}\left[\left\Vert \mu\right\Vert _{2}^{2} \right]  - \left\Vert \mu_0\right\Vert _{2}^{2} \right] \label{eq: best direction eq2}
\end{align}
where the inequality follows from the Courant–Fischer–Weyl min-max principle, with equality if every market $\mu'$ in the support of $\sigma$ is an eigenvector associated with $\overline{\lambda}(\mu_0)$.

Now consider any segmentation $\sigma$.  From \cref{eq: best direction eq1}, notice that if we divide the difference in values $V^{\alpha}(\sigma_{\epsilon}) - W^{\alpha}(\mu_0)$ by $\varepsilon^2$ instead of $\varepsilon$, the limit is zero,
\begin{align*}
    \lim_{\varepsilon\rightarrow0}\frac{V^{\alpha}(\sigma_{\epsilon}) - W^{\alpha}(\mu_0)}{\varepsilon} = 0.
\end{align*}
As a result, if a segmentation achieves the upper bound \cref{eq: best direction eq2} with equality, there for any other segmentation $\sigma'$ of the same magnitude, $\mathbb{E}_{\sigma}\left[\left\Vert \mu\right\Vert _{2}^{2}\right] = \mathbb{E}_{\sigma'}\left[\left\Vert \mu\right\Vert _{2}^{2}\right]$, there exists $\varepsilon$ small enough such that 
\begin{align*}
    V^{\alpha}(\sigma_{\epsilon}) - W^{\alpha}(\mu_0) \geq V^{\alpha}(\sigma'_{\epsilon}) - W^{\alpha}(\mu_0),
\end{align*}
and therefore, $V^{\alpha}(\sigma_{\epsilon}) \geq V^{\alpha}(\sigma'_{\epsilon})$.

%\end{appendix}

%%%%%%%%%%%%%%%%%%%%%%%%%%%%%%%%%%%%%%%%%%%%%%
%%%%%%%%%%%%%%%%%%%%%%%%%%%%%%%%%%%%%%%%%%%%%%
\newpage
	\begin{center}
		{\Huge \bf Online Appendix}
	\end{center}

\section{Proofs of Corollaries} 

\subsection{Proof of \autoref{corol: monotonicity in alpha}}
\begin{proof}

Separating weighted surplus into consumer surplus and revenue, the expression from \cref{eq:main} is equivalent to
\begin{align*}
	&\alpha \left[CS(p,\theta_H) - CS(p,\theta_L) +  \frac{-\frac{R_p(p,\theta_L)}{R_p(p,\theta_H)}CS_p(p,\theta_H) + CS_p(p,\theta_L)}{-\frac{R_p(p,\theta_L)}{R_p(p,\theta_H)}R_{pp}(p,\theta_H) + R_{pp}(p,\theta_L)}\left(R_p(p,\theta_H) - R_p(p,\theta_L)\right)\right] \\
       +& (1-\alpha) \left[R(p,\theta_H) - R(p,\theta_L)\right].
\end{align*}
The above expression is a convex combination of two functions with weights $\alpha$ and $1-\alpha$.  The second function $R(p,\theta_H) - R(p,\theta_L)$ is increasing over $(p(\theta_L),p(\theta_H))$ because $R(p,\theta_H)$ is increasing for any price below the optimal monopoly price $p(\theta_H)$ for type $\theta_H$ and $R(p,\theta_L)$ is decreasing for any price above the optimal monopoly price $p(\theta_L)$ for type $\theta_L$.  Now suppose the convex combination is decreasing, so $\alpha$-IMB holds, and consider a larger weight $\alpha' \geq \alpha$.  Because the convex combination is decreasing but the second term is increasing, the first term must be decreasing.  So if we replace $\alpha$ with the larger $\alpha'$, we are increasing the weight of the decreasing function and decreasing the weight of the increasing function, so the resulting combination must be decreasing as well, implying information is monotonically bad for $\alpha'$.  A similar argument applies for the case of monotonically good information.
\end{proof}

%%%%%%%%%%%%%%%%%%%%%%%%%%%%%%%%%%%%%%%%%%%%%

\subsection{Proof of \autoref{cor: sufficient conditions for IMB-IMG}}\label{app: sufficients}
\begin{proof}
	Consider each of the terms in \cref{eq: the three effects}.
	
	The first term corresponds to the within-type price change effect.  It is
	\begin{align*}
		(p_{\mu}(\mu))^2 \bigg[(1-\mu)V_{pp}^{\alpha}(p(\mu),\theta_L)  + \mu V_{pp}^{\alpha}(p(\mu),\theta_H)\bigg],
	\end{align*}
	which is positive (negative) whenever 
	\begin{align}
		(1-\mu)V_{pp}^{\alpha}(p(\mu),\theta_L)  + \mu V_{pp}^{\alpha}(p(\mu),\theta_H) \label{eq: cor 2 first term}
	\end{align}
	is positive (negative).  A sufficient condition is that $V^{\alpha}$ is convex (concave), as claimed.
	
	The second term corresponds to the cross-types price change effect.  It is
	\begin{align*}
		2 p_{\mu}(\mu) & \bigg[V^{\alpha}_{p}(p(\mu),\theta_H)  - V^{\alpha}_{p}(p(\mu),\theta_L)\bigg],
	\end{align*}
	which is positive (negative) whenever 
	\begin{align*}
		V^{\alpha}_{p}(p(\mu),\theta_H)  - V^{\alpha}_{p}(p(\mu),\theta_L)
	\end{align*}
	is positive (negative) because $p(\mu)$ is increasing in $\mu$.
	
	The third term corresponds to the price curvature effect.  It is
	\begin{align*}
		p_{\mu\mu}(\mu) & \bigg[(1-\mu)V_p^{\alpha}(p(\mu),\theta_L)  + \mu V_p^{\alpha}(p(\mu),\theta_H)\bigg].
	\end{align*}
	Notice that from the definition of $V^{\alpha}$,
	\begin{align*}
		(1-\mu)V_p^{\alpha}(p(\mu),\theta_L)  + \mu V_p^{\alpha}(p(\mu),\theta_H) & = \alpha \bigg( (1-\mu)CS_p(p(\mu),\theta_L)  + \mu CS_p(p(\mu),\theta_H) \bigg) \\
		&+ (1-\alpha) \bigg((1-\mu)R_p(p(\mu),\theta_L)  + \mu R_p(p(\mu),\theta_H)\bigg) \\
		&= \alpha \bigg( (1-\mu)CS_p(p(\mu),\theta_L)  + \mu CS_p(p(\mu),\theta_H) \bigg), \\
		&\leq 0,
	\end{align*}
	where the second equality follows because of the first-order condition of the seller's profit-maximization problem, and the inequality follows because consumer surplus is decreasing in price.
	So the third term is positive (negative) whenever $p_{\mu\mu}(\mu)$ is negative (positive).
	
	To complete the proof, we relate the curvature of $p$ to the two sufficient conditions in the third bullet of the corollary.  For this, let us take two derivatives of the seller's profit-maximization condition,
	\begin{align*}
		(1-\mu)R_p(p(\mu),\theta_L)  + \mu R_p(p(\mu),\theta_H) = 0.
	\end{align*}
	The first derivative implies
	\begin{align*}
		- R_p(p(\mu),\theta_L)  + R_p(p(\mu),\theta_H)
		+ p_\mu(\mu) \bigg[(1-\mu)R_{pp}(p(\mu),\theta_L)  + \mu R_{pp}(p(\mu),\theta_H)\bigg]= 0.
	\end{align*}
	The second derivative implies
	\begin{eqnarray}
		& (p_{\mu}(\mu))^2 & \bigg[(1-\mu)R_{ppp}(p(\mu),\theta_L)  + \mu R_{ppp}(p(\mu),\theta_H)\bigg] \nonumber \\
		& +  2 p_{\mu}(\mu) & \bigg[R_{pp}(p(\mu),\theta_H)  - R_{pp}(p(\mu),\theta_L)\bigg] \nonumber\\
		& + p_{\mu\mu}(\mu) & \bigg[(1-\mu)R_{pp}(p(\mu),\theta_L)  + \mu R_{pp}(p(\mu),\theta_H)\bigg] = 0.
	\end{eqnarray}
	Note that $p_{\mu\mu}(\mu)$ is multiplied by a negative term because both revenue curves are concave.  So $p_{\mu\mu}(\mu)$ is positive (negative) if the sum of the other two terms are positive (negative).  The first term is positive (negative) whenever 
	\begin{align*}
		(1-\mu)R_{ppp}(p(\mu),\theta_L)  + \mu R_{ppp}(p(\mu),\theta_H)
	\end{align*}
	is positive (negative), which is the case if $R_{ppp}(p,\theta)$ is positive (negative) for all $p$ and both $\theta_L$ and $\theta_H$.  The second term is positive (negative) whenever 
	\begin{align*}
		R_{pp}(p(\mu),\theta_H)  - R_{pp}(p(\mu),\theta_L)
	\end{align*}
	is positive (negative), completing the proof.
\end{proof}

\paragraph{Sufficient conditions for monotonically bad information.}
\autoref{cor: sufficient conditions for IMB-IMG} implies that assuming $R_{ppp}(p,\theta) \geq 0$, that is, the marginal revenue curves are convex, and assuming an additional condition that each $V^{1/2}$ is concave, $\tfrac{1}{2}$-IMB (and therefore $\alpha$-IMB for any $\alpha \geq \tfrac{1}{2}$) holds if the more elastic demand curve, $\theta_L$, has 
\begin{enumerate}
	\item a higher derivative, $D_p(p,\theta_L) \geq D_p(p,\theta_H)$, and,
	\item a more convex revenue curve, $R_{pp}(p,\theta_H) \geq R_{pp}(p,\theta_L)$, or equivalently a more steep marginal revenue curve.
\end{enumerate}

The cross-groups price change effect is negative because $D_p(p,\theta_L) \geq D_p(p,\theta_H)$.  The assumption that $V^{1/2}$ is concave is concave implies that the within-group price change effect is negative.  The conditions that $R_{pp}(p,\theta_H) \geq R_{pp}(p,\theta_L)$ and $R_{ppp}(p,\theta) \geq 0$ ensure the price curvature effect is negative by \autoref{cor: sufficient conditions for IMB-IMG}.

\subsection{Proof of \autoref{prop:simplifiedFmility}}
\begin{proof}

Let us without loss of generality parameterize the family as $\mathcal{D} = \{ a(\theta)(D(p)+b(\theta))\}_\theta$.  The first step is to show that $\mathcal{D}$ can be separated as in the theorem. 

Let $\theta_L,\theta_H$ be the demands in the family with the lowest and highest $b$, that is, $b(\theta_L) = \underline{b}$, $b(\theta_H) = \overline{b})$ (choosing arbitrarily if there are multiple candidates). Define
\begin{align*}
    f_1(\theta) &= \frac{a(\theta)(\overline{b} - b(\theta))}{{a}(\theta_L)(\overline{b}-\underline{b})} \\
    f_2(\theta) &= \frac{a(\theta)(b(\theta) - \underline{b})}{{a}(\theta_H)(\overline{b}-\underline{b})}.
\end{align*}

We need to check that for all $\theta = (a,b)$,
\begin{align*}
    f_1(\theta)D(p,\theta_L) + f_2(\theta)D(p,\theta_H) = D(p,\theta),
\end{align*}
That is,
\begin{align*}
    &f_1(\theta){a}(\theta_L) + f_2(\theta){a}(\theta_H) = a, \\
    &f_1(\theta)a(\theta_L)b(\theta_L) + f_2(\theta)a(\theta_H)b(\theta_H) = ab.
\end{align*}
These two equations hold because
\begin{align*}
    &f_1(\theta){a}(\theta_L) + f_2(\theta){a}(\theta_H) = \frac{a(\theta)}{\overline{b}-\underline{b}}((\overline{b} - b(\theta)) + ({b}(\theta)-\underline{b})) = a.\\
    &f_1(\theta)a(\theta_L)b(\theta_L) + f_2(\theta)a(\theta_H)b(\theta_H)=\frac{a(\theta)}{\overline{b}-\underline{b}}((\overline{b} - b(\theta))\underline{b} + ({b}(\theta)-\underline{b})\overline{b}) = ab.
\end{align*}
The second step is to characterize the welfare-monotonicity properties for the binary family $\{D(p,\theta_L),D(p,\theta_H)\}$.  In order to save on notation, suppose without loss of generality that $a(\theta_H) = b(\theta_H) = 1$, that is, $D(p,\theta_H) = D(p)$, and let $a = a(\theta_L),b=b(\theta_L)$.  We want to show that the family $\{a(D(p)+b),D(p)\}$, satisfies $\alpha$-IMB ($\alpha$-IMG) if and only if there is partial inclusion and
  \begin{align*}
        (2\alpha-1) p + \alpha (\frac{pD'(p)}{R''(p)})
  \end{align*}
  is increasing (decreasing) over $[(R')^{-1}(-{b}),(R')^{-1}(0)]$.

  We have
    \begin{align*}
        &D_1 = a (D_2 + b) \\
        &R_1 = a R_2 + apb \\
        &R'_1 = a R'_2 + ab \\
        &R''_1 = a R''_2 \\
        &V_1 = a V_2 + ab(\alpha \overline{p} + (1-2\alpha)p) \\
        &V'_1 = a V'_2 + ab(1-2\alpha).
    \end{align*}
  So
  \begin{align*}
    &V_2(p) - V_1(p) + \frac{-\frac{R'_1(p)}{R'_2(p)}V'_2 + V'_1}{-\frac{R'_1(p)}{R'_2(p)}R''_2 + R''_1}(R'_1(p) - R'_2(p)) \\
    = & V_2(p) - V_1(p) + \frac{-\frac{a R'_2(p) + ab}{R'_2(p)}V'_2 + a V'_2 + ab(1-2\alpha)}{-\frac{a R'_2 + ab}{R'_2(p)}R''_2 + a R''_2}(R'_2(p)(a-1) + ab) \\
    = & V_2(p)(1-a) - ab(\alpha \overline{p} + (1-2\alpha)p) + \frac{V'_2(p) - (1-2\alpha) R'_2(p)}{R''_2(p)}(R'_2(p)(a-1) + ab)
    \end{align*}
    The derivative of the expression is
    \begin{align*}
        & V'_2(p)(1-a) - ab(1-2\alpha) + \frac{V'_2(p) - (1-2\alpha) R'_2(p)}{R''_2(p)}R''_2(p)(a-1) \\& + (\frac{V'_2(p) - (1-2\alpha) R'_2(p)}{R''_2(p)})' (R'_2(p)(a-1) + ab) \\
        = & -(1-2\alpha) R'_2(p)(a-1) - ab(1-2\alpha) + (\frac{V'_2(p) - (1-2\alpha) R'_2(p)}{R''_2(p)})' (R'_2(p)(a-1) + ab) \\
        = & (R'_2(p)(a-1) + ab)(-(1-2\alpha) + (\frac{V'_2(p) - (1-2\alpha) R'_2(p)}{R''_2(p)})').
    \end{align*}
    Notice that $R'_2(p)(a-1) + ab = R'_1(p) - R'_2(p) \leq 0$,

     so we want to show that
    \begin{align*}
        2\alpha - 1 + (\frac{V'_2(p) - (1-2\alpha) R'_2(p)}{R''_2(p)})' \geq 0.
    \end{align*}
    Notice also that $V'_2(p) - (1-2\alpha) R'_2(p) = \alpha pD'(p)$.  So we want to show that
    \begin{align*}
        2\alpha - 1 + (\frac{\alpha pD'_2(p)}{R''_2(p)})' \geq 0.
    \end{align*}
    In other words we want to show that
    \begin{align*}
        (2\alpha-1) p + \alpha (\frac{pD'_2(p)}{R''_2(p)})
    \end{align*}
    to be increasing in $p$ over $(p(\theta_L),p(\theta_H)) = [(R^{-1}_2)'(-{b}),(R^{-1}_2)'(0)]$, as claimed.

    To see that log-concavity of $f$ is sufficient when $\alpha = 1/2$, note that it is sufficient for 
    \begin{align*}
        \frac{pD'(p)}{R''_2(p)}
    \end{align*}
    to be increasing everywhere, or for 
    \begin{align*}
        \frac{R''_2(p)}{pD'(p)} = \frac{2D'(p) + pD''(p)}{pD'(p)} = \frac{2}{p} + \frac{-D''(p)}{-D'(p)} = \frac{2}{p} + (log -D'(p))' = \frac{2}{p} + (\log f(p))'
    \end{align*}
    to be decreasing.  It is sufficient for $(\log f(p))'$ to be decreasing, that is, $f$ is log-concave.
\end{proof}

\subsubsection{Three Cases for How Information Affects CS and TS}\label{onlineapp: three cases}
Consider
\begin{align*}
	f(p) = \frac{c_1(c_2+c_3p)^{c_4}}{p^2}.
\end{align*}
This density function is positive if $c_1 \geq 0$ and $c_2+c_3p \geq 0$ for all $p \in [\underline{p},\overline{p}]$.  Then
\begin{align*}
	\log p^2f(p) = \log c_1 + {c_4}\log(c_2+c_3p).
\end{align*}
Because $D'(p) = -f(p)$ and $R''(p) = 2D'(p) + pD''(p) = -2f(p) - pf'(p)$, we have
\begin{align*}
	\frac{pD'(p)}{R''(p)} = \frac{pf(p)}{2f(p) + pf'(p)} = (\frac{2}{p} + \frac{f'(p)}{f(p)})^{-1} = (\frac{d}{dp} \log p^2f(p))^{-1} = \frac{c_2 + c_3p}{c_3c_4} = \frac{c_2}{c_3c_4} + \frac{p}{c_4}.
\end{align*}
Therefore, the expression in \cref{eq: simplified expression for aD+b} is
\begin{align*}
	(2\alpha-1)p + \alpha (\frac{pD'(p)}{R''(p)}) = p(2\alpha - 1 + \frac{\alpha}{c_4}) +  \frac{\alpha c_2}{c_3c_4},
\end{align*}
which is increasing (decreasing) whenever the multiplier of $p$ is positive (negative), that is, when
\begin{align*}
	\alpha (2  +\frac{1}{c_4}) \geq (\leq)\ 1.
\end{align*}

Now consider three cases for $c_4$.  When $2  +\frac{1}{c_4} \leq 1$, that is, when $-1 \leq c_4 \leq 0$, then $\alpha(2  +\frac{1}{c_4}) \leq 1$ for all $\alpha$ and therefore information is monotonically good regardless of $\alpha$.  When $2  +\frac{1}{c_4} \geq 1$, that is, when $c_4 \leq -1$ or $c_4 \geq 0$, $\alpha$-IMB ($\alpha$-IMG) holds whenever
\begin{align}
	\alpha \geq (\leq)\ \hat{\alpha}:= \frac{1}{2  +\frac{1}{c_4}} = \frac{c_4}{2c_4 + 1}.\label{eq: threshold on alpha for aD+b example}
\end{align}
This threshold $\hat{\alpha} \in (0,1]$ is such that information is monotonically bad for all $\alpha$ above the threshold, is monotonically good for all $\alpha$ below the threshold, and has no effect on $\hat{\alpha}$-surplus.  When $c_4 \leq -1$, $\hat{\alpha}$ ranges from $\frac{1}{2}$ to $1$.  This means that information is monotonically good for total surplus but bad for consumer surplus.  When $c_4 \geq 0$, $\hat{\alpha}$ ranges from $0$ to $\frac{1}{2}$.  This means that information is monotonically good for both total surplus and consumer surplus.

\paragraph{Concavity of revenue.} We next verify that this example indeed has a concave revenue function.  For this, let us calculate the derivative of the density,
\begin{align*}
	f'(p) = \frac{c_1c_3c_4(c_2+c_3p)^{c_4-1}}{p^2}-\frac{2c_1(c_2+c_3p)^{c_4}}{p^3}.
\end{align*}
Therefore,
\begin{align*}
	R''(p) = -2f - pf' &= \frac{c_1(c_2+c_3p)^{c_4-1}}{p^2}\bigg(-2(c_2+c_3p) -pc_3c_4 +2(c_2+c_3p) \bigg) \\
	&= \frac{-c_1c_3c_4(c_2+c_3p)^{c_4-1}}{p},
\end{align*}
which is negative because $c_1 \geq 0$, $c_2+c_3p \geq 0$ for all $p \in [\underline{p},\overline{p}]$, and $c_3,c_4$ have the same signs.

    \section{Derivation for Examples} \label{OnlineAapp:examples}
%%%%%%%%%%%%%%%%%%%%%%%%%%%%%%%%%%%%%%%%%%%%%

\subsection{\autoref{ex: linear with shifts}}\label{app: linear with shifts}
First consider IMG. In this example, both marginal revenue curves are linear because $R_{ppp}(p,\theta) = 0$.  Additionally, the two marginal revenue curves have the same slope, $-2$.  Therefore, the only condition left to check is $D(p,\theta_1) \geq D(p,\theta_2)$ to ensure that the cross-types price change effect is positive.   This condition is
\begin{align*}
	c_1 - c_2 \geq (a_2 - a_1)p
\end{align*}
for all prices in $[\frac{a_1}{2},\frac{a_2}{2}]$.  Because $a_2 \geq a_1$, the right hand side of the above inequality is increasing in $p$ and takes its highest value at $p = p(\theta_2)=\frac{a_2}{2}$, so the above condition becomes
\begin{align*}
	c_1 - c_2 \geq (a_2 - a_1)\frac{a_2}{2}.
\end{align*}

Now consider IMB. Again, the only conditions to check are regarding the ranking of the derivatives and the technical condition $D_p(p,\theta) + pD_{pp}(p,\theta) \leq 0$ that implies that the within-type price change effect is negative.  The condition $D_p(p,\theta_2) \leq D_p(p,\theta_1)$ becomes
\begin{align*}
	c_1 \leq c_2.
\end{align*} 
The condition $D_p(p,\theta) + pD_{pp}(p,\theta) \leq 0$  becomes
\begin{align*}
	c_i \leq p^2
\end{align*}
for all prices in $[\frac{a_1}{2},\frac{a_2}{2}]$ and each $i\in \{1,2\}$, which given the condition $c_1 \leq c_2$ can be summarized as $c_2 \leq \frac{a_1^2}{4}$.

\subsection{\autoref{example:CES}} \label{app:CES}

For a constant-elasticity (CE) demand curve $D(p, \theta)=(c+p)^{-\theta}$, the optimal monopoly price is given by
\begin{align*}
    p(\theta)=\frac{c}{\theta-1}.
\end{align*}
Thus, we require $c>0$ and $\theta>1$ to ensure that a finite monopoly price exists when the monopolist faces each demand curve. To make the exposition of the proof more intuitive, we define $\theta_1=\theta_H$ and $\theta_2=\theta_L$, which means $\theta_H > \theta_L$, and use $(\theta_H,\theta_L)$ instead of $(\theta_1,\theta_2)$. Thus, the statement of the example is
						\begin{align*}
			|\theta_H-\theta_L| \leq \frac{1}{2} \Rightarrow \alpha-\text{IMB holds for } \alpha\geq \frac{1}{2}  
			\end{align*}

 Let $x=\theta_H-\theta_L$, $0<x\leq \frac{1}{2}$. The derivative of the expression in \cref{{eq:main}} simplifies to
            \begin{align*}
-&\frac{(c+p)^{-\theta_L-x-1} \left((c-\theta_L p+p) (c+p)^x-c+\theta_L p+p x-p\right) }{2 \left(2 c^2-c p (2 \theta_L+x-3)+(\theta_L-1) p^2 (\theta_L+x-1)\right)^2} \times\\
& \Big(2 c^4+8 c^3 p-c^2 p^2 (\theta_L (3 \theta_L+2)+3 \theta_L x+x-11)+c p^3 (2 \theta_L+x-3) (\theta_L (\theta_L+x)-2)\\
& \;\;-(\theta_L-1) p^4 (\theta_L+x-1) (\theta_L (\theta_L+x)-1)\Big)
            \end{align*}
The denominator is positive. Thus, in order to show that the derivative is negative for all $p$, we need to show that the numerator is positive.

%%%%%%%%%%%%%%%%%%%%%%%%%%%%%%%%%%%
%%%%%%%%%%%%%%%%%%%%%%%%%%%%%%%%%%%

\begin{enumerate}
\item $(c+p)^{-(\theta_L+x+1)}\geq 0$: This holds because $c,p \geq 0$.

\item $ (c-\theta_L p+p) (c+p)^x-c+\theta_L p+p x-p\geq 0$.

Simplify the above expression to get
\begin{align*}
&(c-(\theta_L-1)p) (c+p)^x-(c-(\theta_L -1)p)+p x
=& \left((c+p)^x-1\right) (c-(\theta_L-1) p)+px
\end{align*}

As $p(\theta_H)\leq p\leq p(\theta_L)$ and $p(\theta_L)=\frac{c}{\theta_L-1}$, $ c-(\theta_L-1) p\geq 0$. 

If $ \left((c+p)^x-1\right)\geq 0$, we are done. So suppose $\left((c+p)^x-1\right)<0$. In this case, because $c-(\theta_L-1) p$ is maximized at $p=p(\theta_H)=p(\theta_L+x)=\frac{c}{\theta_L+x-1}$, we have
\begin{align}
& \left((c+p)^x-1\right) (c-(\theta_L-1) p)+px \nonumber \\
\geq &
\frac{ cx \left((c+p)^x-1\right)}{\theta_L+x-1}+p x.  \label{eq:intermInt}
\end{align}
The derivative of this expression with respect to $p$ is given by
\begin{align*}
    \frac{x^2 c (c+p)^{x-1}}{\theta_L+x-1}+x,
\end{align*}
which is positive. Thus, \cref{eq:intermInt} is minimized when $p$ takes its minimum value at $p=p(\theta_H)=p(\theta_L+x)$. In this case, \cref{eq:intermInt}  simplifies to
\begin{align*}
\frac{c x \left(c \left(\frac{1}{\theta_L+x-1}+1\right)\right)^x}{\theta_L+x-1} \geq 0.  
\end{align*}
Thus, the minimum value of $ \left((c+p)^x-1\right) (c-(\theta_L-1) p)+p x$ is at least zero, which completes this part of the proof.

\item             
\begin{align}
& \Big(2 c^4+8 c^3 p-c^2 p^2 (\theta_L (3 \theta_L+2)+3 \theta_L x+x-11)+c p^3 (2 \theta_L+x-3) (\theta_L (\theta_L+x)-2)\nonumber\\
& \;\;-(\theta_L-1) p^4 (\theta_L+x-1) (\theta_L (\theta_L+x)-1)\Big)\geq 0 \label{eq:eq3CES}
\end{align}
We show that the derivative of \cref{eq:eq3CES} with respect to $x$ is negative, which means that it takes its lowest value at the highest value of $x$, $x = \frac{1}{2}$. The derivative is given by
\begin{align*}
   p^2 \left(\theta_L^2 p (3 (c+p)-2 p x)+\theta_L (c+p) (2 p x-3 c)-(c+p)^2-2 \theta_L^3 p^2\right).
\end{align*}
Dividing this expression by $p^2$ does not change its sign.   So we want to show that the following expression is negative
\begin{align}
   \theta_L^2 p (3 (c+p)-2 p x)+\theta_L (c+p) (2 p x-3 c)-(c+p)^2-2 \theta_L^3 p^2\leq 0. \label{eq:eq354CES}
\end{align}

% Let us first evaluate \cref{eq:eq354CES} at the minimum and maximum $p$ i.e. $p(\theta_L+x)$ and $p(\theta_L)$. At $p(\theta_L+x)$, the expression is $-\frac{c^2 (\theta_L+x) (2 (\theta_L-1) \theta_L+(\theta_L+1) x)}{(\theta_L+x-1)^2}$ while at $p(\theta_L)$, it is $-\frac{2 c^2 \theta_L^2}{\theta_L-1}$. Thus, the two end points are  negative. Recall that  $p(\theta_L)> p(\theta_L+x)$ for $x>0$. We have
% \begin{align*}
%     &-\frac{2 c^2 \theta_L^2}{\theta_L-1}- \left(-\frac{c^2 (\theta_L+x) (2 (\theta_L-1) \theta_L+(\theta_L+1) x)}{(\theta_L+x-1)^2} \right)\\
%     =&-\frac{c^2 x \left(\theta_L^3-\theta_L+\theta_L^2 x+x\right)}{(\theta_L-1) (\theta_L+x-1)^2}
% = -\frac{c^2 x \left(\theta_L(\theta_L^2-1)+x(\theta_L^2 +1)\right)}{(\theta_L-1) (\theta_L+x-1)^2}\\
% =& -\frac{c^2 x (\theta_L+1) \left(\theta_L(\theta_L-1)+x(\theta_L +1)\right)}{(\theta_L-1) (\theta_L+x-1)^2}<0
% \end{align*}
% Which is the case as $\theta_L>1$, thus every term in the fraction is positive. 

% This implies that \cref{eq:eq354CES} has a lower negative value at its maximum $p$ compared to the (negative) value at its minimum $p$. As such, in order for \cref{eq:eq354CES} to be negative for all $p(\theta_L+x)\leq p \leq p(\theta_L)$ it is sufficient to show that the value is decreasing in $p$--that is, the first derivative of  \cref{eq:eq354CES} with respect to $p$ is negative in this range. 

% For this, it is sufficient to show that the derivative at $p(\theta_L+x)$ is negative and the second derivative is also negative-- derivative starts negative and (weakly) decreases as $p$ increases.

For this, let us first evaluate \cref{eq:eq354CES} at the minimum, i.e. $p(\theta_L+x)$. At $p(\theta_L+x)$, the expression is $-\frac{c^2 (\theta_L+x) (2 (\theta_L-1) \theta_L+(\theta_L+1) x)}{(\theta_L+x-1)^2}$, which is negative.

Next, consider the derivative of \cref{eq:eq354CES} with respect to $p$.  This derivative is given by
\begin{align*}
    &c (\theta_L (3 \theta_L+2 x-3)-2)-2 (\theta_L-1) p (\theta_L(2 \theta_L+2 x-1)-1) \\
    \leq & c (\theta_L (3 \theta_L+2 x-3)-2)-2 (\theta_L-1) p(\theta_L+\frac{1}{2}) (\theta_L(2 \theta_L+2 x-1)-1) \\
     = &   -\frac{c \left(2 \theta_L^3+\theta_L +\theta_L^2 (4 x-3)-6 \theta_L  x+2\right)}{2 \theta_L -1}.
\end{align*}
where the inequality follows because the multiplier of $p$ in the first expression is negative, so the expression takes its largest value at the lowest possible $p$, which is $p = p(\theta_L+x) \geq p(\theta_L+\frac{1}{2}) $.

Consider the term in numerator in parenthesis and rewrite as
\begin{align*}
     \left(2 \theta_L^3+\theta_L +\theta_L^2 (4 x-3)-6 \theta_L  x+2\right)=  \left( (2 \theta_L^3- 3\theta_L^2+\theta_L+2) + 2\theta_L x (2 \theta_L-3) \right).
\end{align*}
We want to show that this term is positive when $\theta_L>1$. Both  $ (2 \theta_L^3- 3\theta_L^2+\theta_L+2)$  and $\theta_L (2 \theta_L-3)$ have a positive derivatives when $\theta_L>1$, so they both achieve their minimum value at $\theta_L=1$. In particular, the minimum value of  $\theta_L (2 \theta_L-3)$ is $-1$,  thus $2\theta_L x (2 \theta_L-3)$ is minimized when $x=\frac{1}{2}$. Putting this together, the minimum value of the parenthesis in the numerator is $1>0$.  As such, the lowest value of the derivative is negative.  

 Next, evaluate \cref{eq:eq3CES}  at $x=\frac{1}{2}$ to get
 \begin{align*}
  \frac{1}{4} (c-(\theta_L-1) p)  \left(8 c^3+8 c^2 (\theta_L+3) p+2 c \left(-2 \theta_L^2+\theta_L+9\right) p^2+\left(4 \theta_L^3-5 \theta_L+2\right) p^3\right) 
 \end{align*}
 To show that the above expression is positive, it is sufficient to show that each of its two constituents are positive

\begin{enumerate}
    \item $(c-(\theta_L-1) p)\geq 0$: as we are only interested in $p(\theta_L+\frac{1}{2})\leq p \leq p(\theta_L)$ and $p(\theta_L)=\frac{c}{\theta_L-1}$, this is true.

\item $\left(8 c^3+8 c^2 (\theta_L+3) p+2 c \left(-2 \theta_L^2+\theta_L+9\right) p^2+\left(4 \theta_L^3-5 \theta_L+2\right) p^3\right) \geq 0$

We will show that this expression is positive at minimum value of $p$, $p=p(\theta_L+\frac{1}{2})$, and that it is increasing in $p.$

First, the value of this expression at $p=p(\theta_H)=p(\theta_L+\frac{1}{2})=\frac{2c}{2\theta_L-1}$ is given by
\begin{align*}
    \frac{16 c^3 (2 \theta_L+1)^2}{(2 \theta_L-1)^2} \geq 0.
\end{align*}

Next, take the derivative with respect to $p$ and simplify to get
\begin{align*}
   & 8 c^2 (\theta_L+3)+4 c \left(-2 \theta_L^2+\theta_L+9\right) p+ 3\left(4 \theta_L^3-5 \theta_L+2\right) p^2\\
=&8 c^2 (\theta_L+3)+36 c p+ (2 \theta_L-1)(3 p^2 \theta_L (2 \theta_L+1) -4 c \theta_L p-6p^2 ). 
\end{align*}
The first two terms are positive. So in order to show that the derivative is positive, we show 
\begin{align*}
3 (2 \theta_L+1) \theta_L p^2-4 c \theta_L p-6 p^2 \geq 0.
\end{align*}
Dividing by $p$
\begin{align*}
& 3 (2 \theta_L+1) \theta_L p-4 c \theta_L-6 p = 6 \left(\theta_L^2-1\right) p+\theta_L (3 p-4 c)
\end{align*}
In order to show that the latter expression is positive, it is sufficient to show that it is positive at minimum value of $p$, $p=p(\theta_H)=p(\theta_L+\frac{1}{2})=\frac{2c}{2\theta_L-1}$:
\begin{align*}
  6 \left(\theta_L^2-1\right) p+\theta_L (3 p-4 c)> & 6 \left(\theta_L^2-1\right) p(\theta_L+\frac{1}{2})+\theta_L (3 p(\theta_L+\frac{1}{2})-4 c) \\
  =&2c\bigg(\theta_L+3(1-\frac{1}{2\theta_L-1})\bigg) \geq 0
\end{align*}
The last line is positive as $\theta_L>1$, which implies $\frac{1}{2\theta_L-1}<1$. 
%%%%%%%%%%%%%%%%%%%%%%%%%%%%%%%%%%%
\end{enumerate}

%%%%%%%%%%%%%%%%%%%%%%%%%%%%%%%%%%%%%
%%%%%%%%%%%%%%%%%%%%%%%%%%%%%%%%%%%%%  

\end{enumerate}

\paragraph{Necessity of $\theta_H \leq \theta_L + \frac{1}{2}$.} For any pair of $(\theta_H, \theta_L$), derivative of the expression in \cref{{eq:main}} at $p(\theta_H)$ is given by
\begin{align*}
    -\frac{(2 \theta_H-2 \theta_L+1) (\theta_H-\theta_L) \left(\frac{c  \theta_H}{\theta_H-1}\right)^{1-\theta_H}}{2 c \theta_H^2}
\end{align*}
As $\theta_H>\theta_L>1$, this expression is always negative.
Furthermore, for $(\theta_H, \theta_L$), derivative of the expression in \cref{{eq:main}} at $p(\theta_L)$ is given by
\begin{align*}
    \frac{(2 \theta_H-2 \theta_L-1) (\theta_H-\theta_L) \left(\frac{c \theta_L}{\theta_L-1}\right)^{1-\theta_H}}{2 c \theta_L^2}
\end{align*}
 As $\theta_L>1$, this expression is strictly increasing in $\theta_H$ and is zero if $\theta_H=\theta_L+\frac{1}{2}$, which completes the proof. 

\paragraph{The concavity comparison of \citet{ACV10}.}  \citet{ACV10} use a condition that is violated in our example.  In particular, the condition requires that the demand curve with a lower monopoly price, $\theta_1$, is less convex in the sense that
\begin{align*}
    \frac{D_{pp}(p,\theta_1)}{D_p(p,\theta_1)} \geq \frac{D_{pp}(p,\theta_2)}{D_p(p,\theta_2)}.
\end{align*}
In our example,
\begin{align*}
    D_p(p,\theta_i) &= -\theta_i (c+p)^{-\theta_i - 1}, \\
    D_{pp}(p,\theta_i) &= \theta_i(\theta_i+1) (c+p)^{-\theta_i - 2},\\
    \frac{D_{pp}(p,\theta_1)}{D_p(p,\theta_1)} &= - \frac{1+\theta_i}{c+p}.
\end{align*}
Therefore, their ranking requires that
\begin{align*}
    \frac{1+\theta_1}{c+p} \leq \frac{1+\theta_2}{c+p},
\end{align*}
which is violated because $\theta_1 > \theta_2$.

\end{appendix}

\end{document}